\documentclass[10pt,conference]{IEEEtran} 
\IEEEoverridecommandlockouts

\usepackage{cite}
\usepackage{booktabs}   
\usepackage{subcaption} 
\usepackage[normalem]{ulem}
\usepackage{pifont}
\usepackage{amsmath,amssymb,amsfonts}
\usepackage{algorithm}
\usepackage{algpseudocode}
\usepackage[hyphens,spaces,obeyspaces]{url}
\usepackage{comment}
\usepackage{color}
\usepackage{balance}
\usepackage{listings}
\usepackage{multirow}
\usepackage{enumitem}
\usepackage{mathtools}
\usepackage{graphicx}
\usepackage{pbox}
\usepackage{makecell}
\usepackage{rotating}
\usepackage{adjustbox}
\usepackage{amsthm}
\usepackage{mathtools}
\usepackage{xcolor}
\usepackage{textcomp}
\usepackage[compatibility=false,font=small,labelfont=bf]{caption}

\def\BibTeX{{\rm B\kern-.05em{\sc i\kern-.025em b}\kern-.08em
		T\kern-.1667em\lower.7ex\hbox{E}\kern-.125emX}}

\newcommand{\cmark}{\ding{51}}

\newcommand{\blue}[1]{\textcolor{blue}{#1}}

\newcommand{\loadspy}[0]{\mbox{\textsc{LoadSpy}}}

\newcommand*\rot{\rotatebox{90}} 

\theoremstyle{}
\newtheorem{definition}{Definition}
\newtheorem{claim}{Claim}

\newtheorem{observation}{Observation}

\begin{document}

\title{Redundant Loads: A Software Inefficiency Indicator}

\author{\IEEEauthorblockN{Pengfei Su, Shasha Wen}
\IEEEauthorblockA{
\textit{College of William \& Mary}\\
\{psu, swen\}@email.wm.edu}
\and
\IEEEauthorblockN{Hailong Yang}
\IEEEauthorblockA{
\textit{Beihang University}\\
hailong.yang@buaa.edu.cn}
\and
\IEEEauthorblockN{Milind Chabbi}
\IEEEauthorblockA{
\textit{Scalable Machines Research}\\
milind@scalablemachines.org}
\and
\IEEEauthorblockN{Xu Liu}
\IEEEauthorblockA{
\textit{College of William \& Mary}\\
xl10@cs.wm.edu}
}

\maketitle

\begin{abstract}
Modern software packages have become increasingly complex with millions of lines of code and references to many external libraries. 
Redundant operations are a common performance limiter in these code bases.
Missed compiler optimization opportunities, inappropriate data structure and algorithm choices, and developers' inattention to performance are some common reasons for the existence of redundant operations. 
Developers mainly depend on compilers to eliminate redundant operations. 
However, compilers' static analysis often misses optimization opportunities due to ambiguities and limited analysis scope; automatic optimizations to algorithmic and data structural problems are out of scope. 

We develop \loadspy{}, a whole-program profiler to pinpoint redundant {\em memory load} operations, which are often a symptom of many redundant operations. 
The strength of \loadspy{} exists in identifying and quantifying redundant load operations in programs and associating the redundancies with program execution contexts and scopes to focus developers' attention on problematic code.
\loadspy{} works on fully optimized binaries, adopts various optimization techniques to reduce its overhead, and provides a rich graphic user interface, which make it a complete developer tool. 
Applying \loadspy{} showed that a large fraction of redundant loads is common in modern software packages despite highest levels of automatic compiler optimizations. 
Guided by \loadspy{}, we optimize several well-known benchmarks and real-world applications, yielding significant speedups.
\end{abstract}

\begin{IEEEkeywords}
Whole-program profiling, Software optimization, Performance measurement, Tools.
\end{IEEEkeywords}

\section{Introduction}
\label{sec:introduction}

Production software packages have become increasingly complex. 
They are comprised of a large amount of source code, sophisticated control and data flow, a hierarchy of component libraries, and growing levels of abstractions. 
This complexity often introduces inefficiencies across the software stacks, leading to resource wastage, performance degradation, and energy dissipation~\cite{Molyneaux:2009:AAP:1550832, Bryant:2010:CSP:1841497}.
Such inefficiencies are usually in the form of useless or redundant operations, such as computations whose results may not be used~\cite{Butts:2002:DDD:605397.605419, journals/jilp/SengT05}, re-computation of already computed values~\cite{RVN}, unnecessary data movement~\cite{chabbi2012deadspy, 854389, 6557169, Marin-sweep3d, redspy}, and excessive synchronization~\cite{Chabbi:2015:BEP:2688500.2688502, Tallent:2010:ALC:1837853.1693489}. 
The provenance of these inefficiencies can be many: rigid abstraction boundaries, missed opportunities to optimize common cases, suboptimal algorithm design, inappropriate data structure selection, and poor compiler code generation. 

There is a long history of compiler optimizations aimed at statically analyzing and eliminating redundant operations by techniques such as common sub-expression elimination~\cite{deitz2001eliminating}, value numbering~\cite{gvn},  constant propagation~\cite{Wegman:1991:CPC:103135.103136},  to name a few. 
However, they have a myopic view of the program, which limits their analysis to a small scope---individual functions or files. 
Layers of abstractions, dynamically loaded libraries, multi-lingual components, aggregate types, aliasing, sophisticated flows of control, input-specific path-specific redundancies, and combinatorial explosion of execution paths make it practically impossible for compilers to obtain a holistic view of an application to eliminate all redundancies. 
Link-time optimization~\cite{Fernandez:1995:SEL:207110.207121} can offer better visibility, however, the analysis is still conservative and may err on the side of being less exhaustive to reduce prohibitive analysis cost.
Whole-program link-time optimizations~\cite{Johnson:2017:TSI:3049832.3049845, citeulike:481261} have provided less than 5\% average speedup, although a lot more headroom exists as we show in our work.
Thus, despite their best efforts, compilers often fall short of eliminating runtime inefficiencies. 

Execution profiling aims to understand the runtime behavior of a program.
Performance analysis tools such as HPCToolkit~\cite{adhianto2010hpctoolkit}, VTune~\cite{vtune}, perf~\cite{perf}, gprof~\cite{Graham-etal:1982:PLDI-gprof}, OProfile~\cite{Levon:OProfile}, and CrayPAT~\cite{DeRose-etal:2008:CrayPAT} monitor code execution to identify hot code regions, idle CPU cycles,  arithmetic intensity, and cache misses, etc. 
These tools can recognize the utilization (saturation or underutilization) of hardware resources, but they cannot inform whether a resource is being used in a \emph{fruitful} manner that contributes to the overall efficiency of a program. 
A hotspot need not mean inefficient code, and conversely, the lack of a hotspot need not mean better code.
Coarse-grained profilers usually cannot distinguish efficient vs. inefficient code; for example, they cannot identify that repeated memory loads of the same value or result-equivalent computations waste both memory bandwidth and processor functional units.

 \emph{Whole-program fine-grained monitoring} is a means to monitor execution at microscopic details: it monitors each binary instruction instance, including its operator, operands, and runtime values in registers and memory. 
A key advantage of microscopic program-wide monitoring is that it can identify redundancies irrespective of the user-level program abstractions. 
Prior work~\cite{chabbi2012deadspy, RVN, redspy, toddler} has shown that the fine-grained profiling techniques can identify many forms of software inefficiencies and offer detailed guidance to tune code. 

Existing fine-grained profilers pinpoint inefficiencies in a subset of individual operations such as operations with symbolic equivalence~\cite{RVN}, dead memory stores~\cite{chabbi2012deadspy}, and operations writing same values to target registers or memory locations~\cite{redspy}. 
They have, however, overlooked an important category \emph{temporal load redundancy}---loading the same value from the same memory location. 
For instance, the code on the left of Listing~\ref{lst:example} shows redundant operations that are invisible in existing fine-grained profilers. 
In this code, suppose all the scalars are in registers and vectors are in memory. 
Because there are no ``dead store'' operations (a store followed by another store to the same location without an intervening load), DeadSpy~\cite{chabbi2012deadspy} does not identify any inefficiency.
Since the values written in $t$ and $delta$ always change, RedSpy~\cite{redspy} does not report any ``silent store'' operations~\cite{854389}. 
Finally, since there is no symbolic equivalent computation, RVN~\cite{RVN} does not report any inefficiency. 
Furthermore, because the optimization involves the mathematically equivalent transformation, as shown on the right of Listing~\ref{lst:example}, it is difficult to optimize with other compiler techniques such as polyhedral optimization~\cite{POP2006GRAPHITE}.

\begin{figure}
\begin{minipage}[t]{.49\linewidth}
\begin{lstlisting}[firstnumber=1,language=c]
while (t < threshold) {
  t = 0;
  for(i = 0; i < N; i++) 
@$\blacktriangleright$@   t += A[i] + B[i]*delta;
  delta -= 0.1 * t;
}
\end{lstlisting}
\end{minipage}\hfill
\begin{minipage}[t]{.47\linewidth}
\begin{lstlisting}[firstnumber=1,language=c]
for (i = 0; i < N; i++)
  a += A[i];  b += B[i];
while (t < threshold) {
  t = a + b * delta;
  delta -= 0.1 * t;
}
\end{lstlisting}
\end{minipage}
\vspace{-0.3in}
\captionof{lstlisting}{An example code (on the left) with temporal inefficiencies that cannot be identified by existing fine-grained profilers. 
Because arrays $A$ and $B$ are immutable in the loop nest, computing on these loop invariants introduces many redundancies. One can hoist the redundant computation outside of the loop (on the right) for optimization.
}
\vspace{-0.1in}
\label{lst:example}
\end{figure}

The  code on the left of Listing~\ref{lst:exampleSpatial} shows another kind of load redundancy, which loads the same value from the \emph{nearby} memory locations. Even though each element of array $A$ is only loaded once, adjacent elements with the same values result in loading the same value and the subsequent redundant computation. 
We refer to this type of redundancy as \emph{spatial load redundancy}. As a practical example, a sparse matrix with a dense format can yield many spatial load redundancies.

\begin{figure}
\begin{minipage}[t]{.48\linewidth}
\begin{lstlisting}[firstnumber=1,language=c]
int A[N] = {1, 1, 1, 15}; 
for(i = 0; i < N; i++) 
{
@$\blacktriangleright$@ t += func(A[i]);
}
\end{lstlisting}
\end{minipage}\hfill
\begin{minipage}[t]{.48\linewidth}
\begin{lstlisting}[firstnumber=1,language=c]
int A[N] = {1, 1, 1, 15}; 
a = func(A[0]);
for(i = 0; i < N; i++) {
  if (A[i] != A[i-1]) 
    a = func(A[i]);
  t += a; }
\end{lstlisting}
\end{minipage}
\vspace{-0.3in}
\captionof{lstlisting}{An example code (on the left) with spatial inefficiencies that cannot be identified by existing fine-grained profilers. 
The load redundancy happens at line 4 where the program reads the same value from the nearby memory locations since some adjacent elements of array $A$ have the same value. Such redundancy further results in redundant computation involved in the function $func$. Because $func$ always returns the same value for the same input. One can compare if the adjacent elements in array $A$ are equivalent to eliminate redundant computation (on the right). If they are the same, one can reuse the return value of $func$, which is generated in the previous iteration.
}
\vspace{-.2in}
\label{lst:exampleSpatial}
\end{figure}

Listing~\ref{lst:example}  and~\ref{lst:exampleSpatial} show a tip of the iceberg of the inefficiencies we target in this paper to complement existing tools. 
From our observation, a variety of inefficiencies exhibit \emph{substantial} redundant loads; conversely, the presence of a large fraction of redundant loads in an execution is a symptom of some kind of inefficiency \emph{in the code regions} that exhibit such redundancy.
Furthermore, the subsequent operations based on redundant loads are potentially redundant. 

We have designed and implemented a developer tool---\loadspy{}---aimed at profiling an execution and quantifying load redundancy in the execution.
 \loadspy{} highlights precise source code in its full calling contexts and the two parties involved in a redundant load.
 Additionally, \loadspy{} narrows down the investigation scope to help developers focus on the provenance of inefficiencies.
 A thorough evaluation on a suite of benchmarks and real-world applications shows that looking for redundant loads in a program offers an easy avenue for performance enhancement in many programs.


In this paper~\footnote{This is a full-version of our ICSE paper~\cite{loadspy}.}, we make the following contributions:
\begin{itemize}[leftmargin=*]
\item Show that redundant loads are a common indicator of various forms of software inefficiencies. This finding serves as the foundation of \loadspy{}.
\item Describe the design of \loadspy{}---a whole-program fine-grained profiler to pinpoint redundant loads.
\item Develop strategies for analyzing a large volume of profiling data by attributing redundancy to runtime contexts, objects, and scopes.
\item Enable rich visualization for a large volume of profiling data coming from different threads/processes with a user-friendly GUI, which improves the usability for non-experts.
\item Apply \loadspy{} to pinpoint inefficiencies in well-known benchmarks and real-world applications that were the subjects of study and optimization for years and eliminate \loadspy{}-found inefficiencies by avoiding redundant loads, which yield nontrivial speedups.
\end{itemize}

\section{Related Work}
\label{sec:related}

There exist many compiler techniques and static analysis techniques~\cite{cooper2008redundancy,deitz2001eliminating,Luo:2014:OSC:2628071.2628121,hundt2011mao} to identify redundant computation. 
However, these static approaches suffer from limitations related to the precision of alias information, optimization scope, and insensitivity to inputs and execution contexts. To address these issues, recent approaches convert the source code to specific notations for redundancy detection and removal~\cite{Ding:2017:GGL:3152284.3133898}, or target specific algorithm for optimization~\cite{Ding:2017:GTD:3062341.3062377}. However, these approaches require substantial prior knowledge to identify whether a program suffers from redundancies that are worthy of optimization.
In contrast, \loadspy{} monitors execution, avoids inaccuracies associated with compile-time analysis, and needs no prior knowledge of the measured programs. 

There exist many hardware-based approaches~\cite{Lipasti:1996:VLL:237090.237173,Lipasti:1996:EDL:243846.243889,854389,Lepak:2000:SSF:360128.360133,miguel2014load, miguel2015doppelganger,yazdanbakhsh2016rfvp,Butts:2002:DDD:605397.605419} that optimize redundant operations during program execution. However, these approaches require hardware extension, which is unavailable in commodity processors. 
Instead, \loadspy{} is a pure software approach and does not need any hardware changes. 
The remaining section reviews only other profiling techniques. 

\subsection{Value profiling}
\loadspy{} is a value-aware profiler; value profiling techniques are closely related to our work. 
Calder et al.~\cite{Calder:1997:VP:266800.266825,Calder99valueprofiling,Feller98valueprofiling} proposed probably the first value profiler on DEC Alpha processors. 
They instrumented the program code and recorded top N values to pinpoint invariant or semi-invariant variables stored in registers or memory. 
A variant of this value profiler was proposed in a later research~\cite{Watterson:2001:GVP:647477.760386}. 
Burrows et al.~\cite{Burrows:2000:EFV:378993.379236} used hardware performance counters to sample values in Digital Continuous Profiling Infrastructure~\cite{DCPI}. 
Wen et al.~\cite{witch} combined performance monitoring units and debug registers available in x86 to identify redundant memory operations. These approaches do not explore whole-program load redundancy in depth. 
Moreover, none of them detect spatial redundancy.


Some code specialization work depends on value profiling. 
However, these approaches limit themselves to only analyzing registers~\cite{Muth:2000:CSB:647169.718148}, static instructions~\cite{Oh:2013:PAL:2451116.2451161}, memory store operations~\cite{redspy}, or functions~\cite{Chung00energyefficient,Kamio04avalue,vprof}. 
They omit many optimization opportunities and require significant manual efforts to reason about the root causes of inefficiencies.


Unlike existing value profilers, \loadspy{} has four distinct features. 
First, \loadspy{} is the first value profiler that tracks the \emph{history of loaded values} from individual \emph{memory locations}, rather than the values produced by \emph{individual instructions}.
Second, \loadspy{} identifies both \emph{temporal and spatial} redundancies in load operations. 
Third, \loadspy{} provides novel redundancy scope and metrics to guide optimization in both contexts and semantics. 
Fourth, \loadspy{} not only identifies redundancy arising due to exactly the same values but also identifies redundancy due to approximately equal values, which offers opportunities for \emph{approximate computing}.

\subsection{Value-agnostic profiling}
RVN~\cite{RVN} assigns symbolic values to dynamic instructions and identifies redundancy on the fly. DeadSpy~\cite{chabbi2012deadspy} tracks every memory operation to pinpoint a store operation that is not loaded before a subsequent store to the same location. 
MemoizeIt~\cite{DellaToffola:2015:PPY:2814270.2814290} detects Java methods that perform identical computations. Travioli~\cite{Padhye:2017:TRD:3097368.3097425} detects redundant data structure traversals.
These approaches miss out on certain opportunities that \loadspy{} can detect by explicitly inspecting values generated at runtime.


Toddler~\cite{toddler} has to manually add loop events to instrument loops in a C code base and only identifies repetitive memory loads across loop iterations.
The follow-on work LDoctor~\cite{ldoctor} reduces Toddler's overhead using a combination of ad-hoc sampling and static analysis techniques. 
LDoctor  instruments a small number of suspicious loops at compile time. 
This technique can miss redundant loads in different loops. 
In contrast, \loadspy{} works on fully optimized binaries, is independent of any compiler, and performs the whole-program profiling instead of limiting itself to only profiling loops.


\section{Redundant Loads: An Inefficiency Symptom}
\label{sec:motivation}


While there are several ways to identify the inefficiency, \loadspy{} focuses on memory load operations. 
If two consecutive load operations performed on the same memory location load the same value, the second load operation can be deemed useless. 
Thus, the second load could potentially be elided. 
Our study aims to quantify redundant loads and attribute them to the code regions that cause them. 
\emph{A single instance of a redundant load is uninteresting; highly frequent redundant loads occurring in the same code location demand attention.}

It is easy to imagine how redundant loads happen: repeatedly accessing immutable data structures or algorithms employing memoization.
It is equally easy to see how inefficient code sequences show up as redundant loads: missed inlining appears as repeatedly loading the same values in a callee, imperfect alias information shows up as loading the same values from the same location via two different pointers, redundant computations show up as the same computations being performed by loading unchanged values, algorithmic defects, e.g., frequent linear searches or hash collisions, also appear as repeatedly loading unchanged values from the same locations. 

\begin{definition}[Temporal Load Redundancy]
A memory load operation $L_2$, loading value $V_2$ from location $M$, is redundant $iff$ the previous load operation $L_1$,  performed on $M$, loaded a value $V_1$ and $V_1 = V_2$. If $V_1 \approx V_2$, we call it approximate temporal load redundancy.
\end{definition}

\begin{definition}[Spatial Load Redundancy]
A memory load operation $L_2$, loading a value $V_2$ from location $M_2$, is redundant $iff$ the previous load operation $L_1$,  performed on location $M_1$, loaded a value $V_1$ and $V_1 = V_2$, and $M_1$ and $M_2$ belong to the address span of the same data object. If $V_1 \approx V_2$, we call it approximate spatial load redundancy.
\end{definition}

\begin{definition}[Redundancy Fraction]
We define the \emph{redundancy fraction} ${\mathcal R}$ in an execution as the ratio of bytes redundantly loaded to the total bytes loaded in the entire execution. 
\end{definition}

We emphasize that the redundancy is defined for instruction instances, \emph{not} static instructions.
Deleting an instruction involved in one instance of a redundant load can be unsafe.

\begin{observation}
Large redundancy fraction ($\mathcal R$) in the execution profile of a program is a symptom of some kind of software inefficiency.
\end{observation}

Redundant loads are neither a necessary condition nor a sufficient condition to capture all kinds of software inefficiencies.
However, we show, with many illustrative case studies, that \emph{a large fraction of redundant loads in the same code region} is often a symptom of a serious inefficiency. 
We notice frequent redundant loads across the board in many programs irrespective of optimization levels, raising a warning alarm of potential inefficiency. Although not all redundant loads demand optimization, in our experience, investigating the top few contributors in a profile offers a high potential to tune and optimize code. Looking for load redundancy opens potentially an easy avenue for code optimization---manual or automatic.  

We measure the redundancy fraction in a number of benchmarks SPEC CPU2006~\cite{SPEC:CPU2006}, PARSEC-2.1~\cite{parsec}, Rodinia-3.1~\cite{rodinia}, and NERSC-8~\cite{TRINITY-WWW}. 
We compile these benchmarks with \texttt{gcc-4.8.5 -O3}, link-time optimization (LTO) and profile-guided optimization (PGO), which is one of the highest optimization levels.
In practice, most packages do not use this level of optimization.

We observe that a large load redundancy fraction correlates with some kind of inefficiency. 
Furthermore, the code that generates many redundant loads is responsible for the inefficiencies in the program.
We classify the causes of redundant loads according to their provenance: input-sensitive redundant loads, inefficient data structure/algorithm designs, or missing compiler optimizations. 
Different kinds of inefficiencies require different optimization strategies.  

\subsection{Input-sensitive Redundant Loads}
\label{subsec:input}

In this section, we classify the inefficiency due to inputs.
Rodinia-3.1 backprop~\cite{rodinia}, a supervised machine learning algorithm, trains the weights of connections in a neural network. The redundancy fraction of this program is 64\%.
It is common knowledge that as the training progresses, many weights stabilize and do not change. Hence, their gradients become and remain zero.
Listing~\ref{lst:backprop} shows the inefficiency at line 3, where the majority of elements in arrays \texttt{delta} and \texttt{oldw} are zeros. 
Computations at lines 3-5 can be bypassed when \texttt{delta[j]} and \texttt{oldw[k][j]} are zeros. 
Repeatedly loading the zero value from \texttt{delta[j]} and \texttt{oldw[k][j]} shows up as spatial load redundancy.
It is easy to eliminate the input-sensitive redundant loads by predicating the subsequent computation on the values of \texttt{delta[j]} and \texttt{oldw[k][j]} being non-zero.

\begin{figure}[t]
\begin{lstlisting}[firstnumber=1,language=c, caption= {Spatial load redundancy in Rodinia-3.1 backprop. Arrays \texttt{delta} and \texttt{oldw} are repeatedly loaded from memory whereas most array elements are zero.}\vspace{-1.5em}, label=lst:backprop]
for (j = 1; j <= ndelta; j++) {
  for (k = 0; k <= nly; k++) {
@$\blacktriangleright$@   new_dw = ((ETA*delta[j]*ly[k])+(MOMENTUM*oldw[k][j]));
     w[k][j] += new_dw;
     oldw[k][j] = new_dw;
  }}
\end{lstlisting}
\end{figure}

\subsection{Redundant Loads due to Suboptimal Data Structures and Algorithms}
\label{subsec:algorithm}

Inefficiencies of this category require semantics to identify and optimize. These inefficiencies also incur a significant number of redundant loads. We illustrate some algorithms that introduce inefficiencies in a few well-known benchmarks.

\paragraph{\textbf{Linear search}}
Rodinia-3.1 particlefilter~\cite{rodinia} is used to estimate the location of a target object in signal processing and neuroscience. The redundancy fraction of this program is 99\%. Listing~\ref{lst:particlefilter} shows the inefficiency in function \texttt{findIndex}, which performs a linear search (line 3) over a sorted array \texttt{CDF} to determine the location of a given particle. 
This linear search is called multiple times in a loop to become the bottleneck of the program. 
The symptom of this inefficiency is many redundant loads, which is caused by the repeated loads of immutable array \texttt{CDF} elements in different invocation instances of function \texttt{findIndex}. To fix this problem, one can replace the linear search with a binary search, which reduces the volume of redundant loads.

\begin{figure}
\begin{lstlisting}[firstnumber=1,language=c, caption= Temporal load redundancy in Rodinia-3.1 particlefilter. A linear
search loads the same values from the same memory locations.\vspace{-1em}, label=lst:particlefilter]
int findIndex(double *CDF, int lengthCDF, double value) {
  for(x = 0; x < lengthCDF; x++) {
@$\blacktriangleright$@   if (CDF[x] >= value) {
      index = x; break;
     }}
  ...
  return index;
}
...
for(j = 0; j < Nparticles; j++)
  i = findIndex(CDF, Nparticles, u[j]); 
\end{lstlisting}
\end{figure}

\paragraph{\textbf{Hash table}}

Parsec-2.1 dedup~\cite{parsec} compresses data via deduplication. The redundancy fraction of this program is 75\%. Listing~\ref{lst:dedup} shows the inefficiency in the program, which searches for an item in a linked list associated with a hash table entry.
The inefficiency comes from the frequent execution on the slow path due to the hash collision. 
We noticed that only $\sim$2\% hash buckets are occupied, and the slow path is frequently taken.
The linked list traversal on the slow path loads the same values from the same locations (line 8), which results in redundant loads.
One can improve the hash function to make hash keys uniformly distributed among buckets, which will reduce the redundancy and hence the inefficiency.
 
\begin{figure}[t]
\begin{lstlisting}[firstnumber=1,language=c, caption= Temporal load redundancy in Parsec-2.1 dedup. Excessive hash collisions in linear hashing result in long linked lists.\vspace{-1.5em}, label=lst:dedup]
struct hash_entry *hashtable_search(struct hashtable *h, void *k) {
  struct hash_entry *e;
  unsigned int hashvalue, index;
  hashvalue = hash(h,k);
  index = indexFor(h->tablelength,hashvalue);
  e = h->table[index];
  while (NULL != e) {
@$\blacktriangleright$@  if ((hashvalue == e->h) && (h->eqfn(k, e->k))) return e;
    e = e->next;  
  } ...}
\end{lstlisting}
\end{figure}

\subsection{Redundant Loads due to Missing Compiler Optimizations}
\label{subsec:compiler}

Inefficiencies of this category occur in small scopes---loop nests or procedure calls.
One needs to either curate the code or manually apply transformations to eliminate these inefficiencies. The following three examples illustrate our findings.

\paragraph{\textbf{Missing scalar replacement}} 

Rodinia-3.1 hotspot 3D~\cite{rodinia} is a thermal simulation program that estimates processor temperature. The redundancy fraction of this program is 95\%. Listing~\ref{lst:hotspot3D} shows a loop nest that performs a stencil computation. 
At line 8, \texttt{tOut\_t[c]} is updated with the values in nearby \texttt{tIn\_t[]}. 
Typically, \texttt{w} $=$ \texttt{c} - \texttt{1} and \texttt{e} $=$ \texttt{c} + \texttt{1}.
As a result, the value of \texttt{tIn\_t[e]} in the current iteration equals the value of \texttt{tIn\_t[c]} in the next iteration and further equals the value of \texttt{tIn\_t[w]} in the iteration after the next. 
However, the compiler does not perform register promotion of \texttt{tln\_[e]}. 
Hence, many \emph{redundant loads} occur in this loop nest. 
To fix this inefficiency, we employ the scalar replacement to eliminate inter-iteration redundant loads from memory. Specifically, we store the value of \texttt{tIn\_t[e]} in a local variable in the current iteration to be reused by \texttt{tIn\_t[c]} in the next iteration and by \texttt{tIn\_t[w]} in the iteration after the next.
 
\begin{figure}[t]
\begin{lstlisting}[firstnumber=1,language=c, caption= Temporal load redundancy in Rodinia-3.1 hotspot3D. Array \texttt{tIn\_t} is repeatedly loaded from memory while the values remain unchanged.\vspace{-1em}, label=lst:hotspot3D]
for(y = 0; y < ny; y++) {
  for(x = 0; x < nx; x++) {
    int c, w, e, n, s, b, t;
    c = x + y * nx + z * nx * ny;
    w = (x == 0) ? c : c - 1;
    e = (x == nx - 1) ? c : c + 1;
    ...
@$\blacktriangleright$@   tOut_t[c] = cc*tIn_t[c]+cw*tIn_t[w]+ce*tIn_t[e]+...
  }}
\end{lstlisting}
\end{figure}

\paragraph{\textbf{Missing constant propagation}} 

NERSC-8 msgrate~\cite{TRINITY-WWW} measures the message passing rate via the MPI interface. The redundancy fraction of this program is 97\%. Listing~\ref{lst:msgrate} shows a procedure \texttt{cache\_invalidate}, which sets all the elements in array \texttt{cache\_buf} to 1. 
This code adopts a suboptimal forward propagation that loads the value of \texttt{cache\_buf[i-1]} and assigns it to \texttt{cache\_buf[i]}. 
Although there is no redundant load in one invocation of this function, procedure \texttt{cache\_invalidate} is called in a loop (not shown in the listing), resulting in excessive, redundant loads from array \texttt{cache\_buf}. 
The compiler does not replace the assignment with a constant, possibly due to its inability to prove the safety of assigning to a global array in the presence of concurrent threads of execution. 


\begin{figure}[t]
\begin{lstlisting}[firstnumber=1,language=c, caption=  Temporal load redundancy in NERSC-8 msgrate. The program repeatedly  loads a constant ``1'' from array \texttt{cache\_buf}.\vspace{-1.5em}, label=lst:msgrate]
int *cache_buf;
...
static void cache_invalidate(void) {
  int i;
  cache_buf[0] = 1;
  for (i = 1; i < cache_size; ++i) 
@$\blacktriangleright$@   cache_buf[i] = cache_buf[i-1];
}
\end{lstlisting} 
\end{figure}

\begin{figure}[t]
\begin{lstlisting}[firstnumber=1, language=c, caption= Temporal load redundancy in SPEC CPU2006 464.h264ref due to missing function inlining.\vspace{-1.5em}, label=lst:h264ref]
for (pos = 0; pos < max_pos; pos++) {
  ...
  if(abs_y >= 0 && abs_y <= max_height && ...) 
    PelYline_11 = FastLine16Y_11; 
  else PelYline_11 = UMVLine16Y_11;
  for (blky = 0; blky < 4; blky++) {
    for (y = 0; y < 4; y++) {
@$\blacktriangleright$@     refptr = PelYline_11(ref_pic, abs_y++, abs_x, img_height, img_width);
      ... 
    } ...}}
\end{lstlisting}
\end{figure}

\paragraph{\textbf{Missing inline substitution}} 

SPEC CPU2006 464.h264ref~\cite{SPEC:CPU2006} is a reference implementation of H.264, a standard of video compression. The redundancy fraction of this program is 84\%. 
The compiler fails to inline the frequently called function \texttt{PelYline\_11} at line 8 shown in Listing~\ref{lst:h264ref}. Because it is invoked via a function pointer and the callee routines are not present in the same file.
The parameters of \texttt{PelYline\_11}---\texttt{abs\_x}, \texttt{img\_height}, and {\texttt{img\_width}---are unmodified across multiple successive invocations. 
In each invocation, the caller pushes the same parameters on the same stack, and then the callee loads the same values from the same location, which show up as redundant loads.
To fix the problem, we need to manually inline the function~\cite{redspy}. 

\paragraph{\textbf{Discussion}}
We have explored other compiler flags that enable advanced optimization such as polyhedral optimization~\cite{graphite-www} in \texttt{GCC}. 
Unfortunately, the polyhedral optimization was unsuccessful in optimizing any of the aforementioned scenarios.
Furthermore, we observed that using LTO, PGO, together with the polyhedral optimization made compilation time extremely high for some cases.
For example, it took over two hours to compile hotspot 3D, a 30,000$\times$ slowdown compared to simply using \texttt{-O3}.
As a result, our later evaluation section  does not use LTO and polyhedral optimization, but only uses \texttt{-O3} with PGO.
We leave the effectiveness of other compilers such as LLVM~\cite{Lattner:2004:LCF:977395.977673} and ICC~\cite{ICC-WWW} on the same set of programs for a future study.

\section{\loadspy{} Implementation}
\label{sec:methodology}

\loadspy{} employs Intel \texttt{Pin}~\cite{Luk:2005:PBC:1065010.1065034} to intercept every memory load operation.
The instrumentation obtains the effective address $M$ to be accessed in the instruction, the access length $\delta$, and offers the pair to a runtime analysis routine. 
In the rest of this section, we discuss how \loadspy{} identifies \emph{temporal} and \emph{spatial} load redundancies, respectively.

\subsection{Detecting Temporal Load Redundancy}
\label{subsec:temporalRed}

Detecting temporal load redundancy requires two pieces of information: the current value $v_{new}$ at the target location and the last-time loaded value $v_{old}$ from the same location.
The runtime analysis routine, run just before the execution of the original program's load instruction, fetches the current value $v_{new}$ at the memory range $[M:M+\delta)$.
\loadspy{} employs a shadow memory $S$ for maintaining the last-time loaded value at the same location.
 $S[M]$ maintains the value last loaded by the program at location $M$.
\loadspy{} utilizes the page-table-based scheme~\cite{chabbi2012deadspy} to efficiently manage its shadow memory. 
At runtime, the analysis routine fetches $v_{old}$ from $S[M:M+\delta)$ and $v_{new}$ from $[M:M+\delta)$.
\loadspy{} records an instance of a \emph{redundant} load if $v_{old} = v_{new}$.
All bytes must match to qualify a load as redundant. 
Intuitively, sub-read-size redundancy is not actionable by the programmer. 
Note, however, that $v_{old}$ might have been generated by multiple shorter reads, a single longer read, or more commonly a single read of the same size. 
If not redundant, \loadspy{} updates the shadow memory with the newly loaded value.
Also, \loadspy{} records an instance of a \emph{non-redundant} load if $v_{old} \ne v_{new}$.

\loadspy{} provisions for approximate computation by allowing the new value generated in a floating-point (FP) operation to \emph{approximately} match the previously present value. 
If the two values are within a threshold of difference, \loadspy{} considers them approximately equal and records an instance of a redundant load.
The threshold is tunable; we use 1\% in our experiments.
Accordingly, \loadspy{} decomposes the load redundancy into \emph{precise} and \emph{approximate}. 

\loadspy{} attributes each instance of redundant loads (and non-redundant loads) to two parties $\langle C_{old}, C_{new}\rangle$ involved in two operations, where $C_{old}$ is the calling context of the previous load operation on $M$ and $C_{new}$ is the calling context of the current load operation on $M$. 

The following equations compute the fraction of temporal load redundancy in an execution:
\begin{eqnarray}
\scriptsize
\begin{aligned}
{\mathcal R}_{prog}^{precise} =& {\sum_i\sum_j\text{Redundant non-FP bytes loaded in } \langle C_{i}, C_{j}\rangle \over \sum_i\sum_j\text{non-FP bytes loaded  in } \langle C_{i}, C_{j}\rangle} \\
{\mathcal R}_{prog}^{approx} =& {\sum_i\sum_j\text{Redundant FP bytes loaded  in } \langle C_{i}, C_{j}\rangle \over \sum_i\sum_j\text{FP bytes loaded  in } \langle C_{i}, C_{j}\rangle} \\
\end{aligned}
\end{eqnarray}

Load redundancy between a pair of calling contexts is given by the following equations:
\begin{eqnarray}
\scriptsize
\begin{aligned}
{\mathcal R}_{\langle C_{old}, C_{new}\rangle}^{precise} =& {\text{Redundant non-FP bytes loaded  in } \langle C_{old}, C_{new}\rangle \over \sum_i\sum_j\text{non-FP bytes loaded  in } \langle C_{i}, C_{j}\rangle}  \\
{\mathcal R}_{\langle C_{old}, C_{new}\rangle}^{approx} =& {\text{Redundant FP bytes loaded  in } \langle C_{old}, C_{new}\rangle \over \sum_i\sum_j\text{FP bytes loaded  in } \langle C_{i}, C_{j}\rangle}  \\
\end{aligned}
\end{eqnarray}
The metrics help identify code regions (pairs of calling contexts) where the highest amount of redundancy is observed.

\textbf{\textit{Obtaining the Runtime Calling Context of an Instruction:}}
\label{subsec:programCtxt}
Attributing runtime statistics to a flat profile (just an instruction pointer) does not offer full insights for developers.
For example, attributing redundant loads to a common library function, e.g., \texttt{strcmp}, offers little insight since \texttt{strcmp} can be invoked from several places in a large code base; some invocations may not even be obvious to the user code.
A detailed attribution demands associating profiles to the full calling context: \texttt{main():line->A():line->...}\texttt{->strcmp():line}.
\loadspy{} requires obtaining the calling context on each load operation since each load---redundant or not.
\loadspy{} employs CCTLib~\cite{Chabbi:2014:CPP:2581122.2544164}, which efficiently maintains calling contexts as a calling context tree (CCT)~\cite{Ammons:1997:EHP:258915.258924} including complex control flows through \texttt{longjump}, tail calls, and exceptions. 
The calling context, which is provided as a unique 32-bit integer, is recorded (in addition to the last-time loaded value) in the shadow memory.

\subsection{Detecting Spatial Load Redundancy}
\label{subsec:spatialRed}

\begin{figure}[t]
\begin{center}
\includegraphics[width=0.45\textwidth]{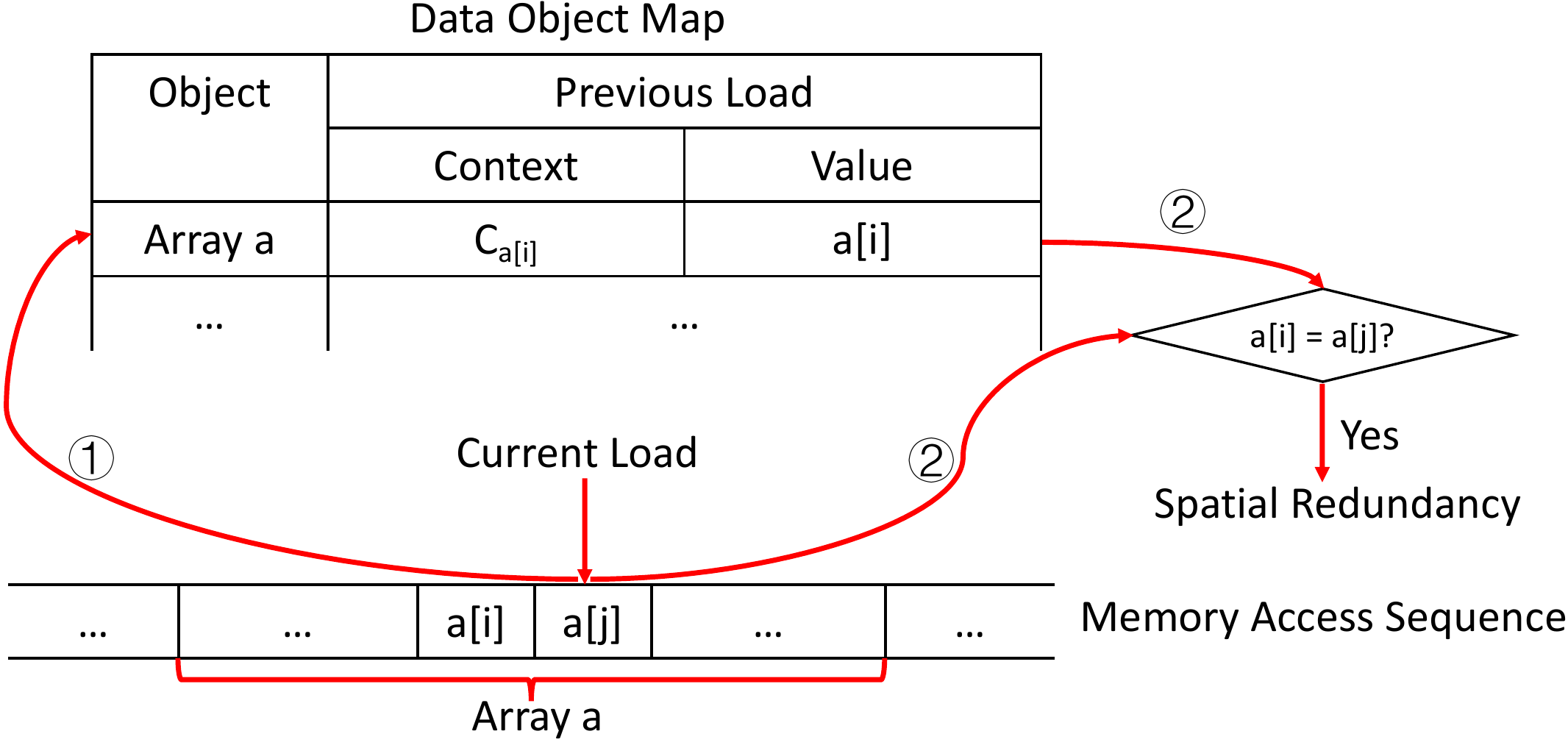}
\end{center}
\vspace{-0.15in}
\caption{Detecting spatial load redundancy. \textcircled{1} \loadspy{} monitors a load operation and associate its effective address with the data object. In a map, each data object associates itself with the value and context of the previous load belonging to this data object. \textcircled{2} \loadspy{} compares the previous and current load values; if they are (approximately) the same, an instance of (approximate) spatial load redundancy is reported. \textcircled{3} The value and context associated with the data object are updated with the ones from the current load.}
\vspace{-1.5em}
\label{fig:spatialRed}
\end{figure}

For arrays and aggregate objects, \loadspy{} checks whether two consecutive loads from any element of the same object load (approximately) the same value. 
For example, if two consecutive loads from an array \texttt{a}, say \texttt{a[i]} and \texttt{a[j]}, load the same value, \loadspy{} flags it as an instance of \emph{spatial} load redundancy and attributes it to the same data object, as shown in Figure~\ref{fig:spatialRed}.

To facilitate spatial load redundancy detection, \loadspy{} maintains a mapping from address ranges to active data objects in a shadow memory. 
Associated with each data object $\mathcal{O}$ is two additional pieces of information: a singleton value $v_{old}$ loaded as a result of the previous load operation performed on $\mathcal{O}$ and the calling context $C_{old}$ associated with the previous load operation performed on $\mathcal{O}$. 
Upon each memory load, \loadspy{} uses the effective address of the load operation to look up the data object it belongs to in the map.
If the value of the current load matches the one recorded with the previous load on the same object, \loadspy{} records an instance of spatial load redundancy. 
The redundancy is hierarchically attributed first to the data object involved and then to the two calling contexts involved in the redundancy.

\loadspy{} provides the similar whole-program and per-redundancy-pair metrics for spatial redundancy. 
Moreover, \loadspy{} computes the per-data-object metrics with the following equations where ${\mathcal O}$ is a data object.
\begin{eqnarray}
\scriptsize
\begin{aligned}
{\mathcal R}_{\mathcal O}^{precise} =& {\text{Redundant non-FP bytes in object }  {\mathcal O}  \over \sum_i\text{non-FP bytes in object i}} \\
{\mathcal R}_{\mathcal O}^{approx} =& {\text{Redundant FP bytes in object }  {\mathcal O}  \over \sum_i\text{FP bytes in object i}} \\
\end{aligned}
\end{eqnarray}

\textbf{\textit{Obtaining Data-object Addresses at Runtime:}}
\loadspy{} monitors static and dynamic data objects but ignores stack objects from spatial redundancy detection.
Data allocated in the \texttt{.bss} section in a load module are static objects. 
Each static object has a named entry in the symbol table that identifies the memory range for the object with an offset from the beginning of the load module. 
The lifetime of static objects begins when the enclosing load module (executable or dynamic library) is loaded into memory and ends when the load module is unloaded. 
\loadspy{} intercepts the loading and unloading of load modules to monitor the lifetime of static data objects and establishes a mapping from an object's address range to the corresponding data object.

Dynamic objects are allocated via one of malloc family of functions (\texttt{malloc}, \texttt{calloc}, \texttt{realloc}) and \texttt{mmap}~\cite{Liu:2013:DPP:2503210.2503297}. 
The memories for dynamic objects are reclaimed at \texttt{free} and \texttt{munmap}. 
\loadspy{} intercepts these functions to establish a mapping from an object's address range to the corresponding data object.
Querying an address at runtime obtains a handle to the corresponding static or dynamic object. The handle is a unique identifier representing the object name for a static object or the allocation calling context for a dynamic object. 
 
\subsection{Identifying the Redundancy Scope}
\label{subsec:redScope}

When the redundancy happens in the same calling context, that is $C_{old} = C_{new}$, there is guaranteed to be a loop~\footnote{We consider natural loops~\cite{Torczon:2007:EC:1526330} only.} around the redundancy location.
However, in code with nested loops, it is unclear whether the redundancy occurred between iterations of an inner loop or between iterations of an outer loop or some other loop in-between.
Hence, it becomes necessary to point out the syntactic scope enclosing a redundancy pair.

We illustrate the need for scope using a real-world application \texttt{MASNUM-2.2}~\cite{Qiao:2016:HEG:3014904.3014911} shown on the left of Listing~\ref{lst:motivationExample}. \loadspy{} identifies 91\% of memory loads are redundant and the top contributor is at line 6.
It is tempting to infer that \texttt{x(iii+1)} loaded in one iteration of the inner \texttt{do} loop (line 5) is loaded again as \texttt{x(iii)} in the next iteration.
An obvious optimization is to perform scalar replacement to retain \texttt{x(iii+1)} across iterations of the inner \texttt{do} loop (on the right of Listing~\ref{lst:motivationExample}).
However, this optimization does not eliminate many redundant loads.
Actually, the outer \texttt{do} loop at line 1 repeatedly searches for an item \texttt{xx}, and the inner \texttt{do} loop performs a linear search.
As a result, the inner loop repeatedly loads the same set of elements across two trips of the outer loop. 
Thus, the load redundancy exists not only between iterations of the inner loop but also between iterations of the outer loop.
The load redundancy at the outer loop highlights an algorithm-level inefficiency---repeated linear searches. 
With this knowledge, we can replace the linear search with a binary search.
More details are shown in \S~\ref{subsec:masnum}.

To assist developers to focus on the \emph{scope} where load redundancy occurs, we have incorporated a \emph{redundancy scope} feature in \loadspy{}. 
We denote redundancy scope with the symbol $\mathcal{S}$.
In Listing~\ref{lst:motivationExample}, the redundancy scope is the \emph{outer} \texttt{do} loop.
Below we detail how redundancy scope is computed.

\begin{figure}
\begin{minipage}[t]{.48\linewidth}
\begin{lstlisting}[firstnumber=1,language=c]
do 500 k=1, kl
  ...
  xx=x0-deltt*(cgx+ux(ia,ic))/rslat(ic)*180./pi
  ...
  do iii = ixs, ixl-1
  @$\blacktriangleright$@ if(xx >= x(iii) .and. xx <= x(iii+1)) then
      ixx = iii; exit
    endif
  enddo
  ...
500 continue
\end{lstlisting}
\end{minipage}\hfill
\begin{minipage}[t]{.48\linewidth}
\begin{lstlisting}[firstnumber=1,language=c]
do 500 k=1, kl
  scalar = x(ixs)
  do iii = ixs, ixl-1
    if(xx >= scalar) then
      scalar = x(iii+1)
      if (xx <= scalar) then 
        ixx = iii; exit
      endif
    else scalar = x(iii+1)
    endif
  enddo
  ...
500 continue
\end{lstlisting}
\end{minipage}
\vspace{-0.3in}
\captionof{lstlisting}{A code example (on the left) from MASNUM-2.2~\cite{Qiao:2016:HEG:3014904.3014911} that requires additional information for disambiguating the scope of load redundancy. 
Many redundant loads occur at line 6 where the array \texttt{x} is repeatedly loaded from memory. If we only focus on the inner loop, we would be misled to believe the stencil computation, which loads \texttt{x(iii+1)} and \texttt{x(iii)}, causes many redundant loads across iterations of the inner loop. However, performing scalar replacement (on the right) does not yield much speedup. An algorithmic-level redundancy happens in the outer \texttt{do} loop, which repeatedly performs linear searches for a sorted array of elements.
}
\vspace{-1em}
\label{lst:motivationExample}
\end{figure}

We first extend calling contexts to incorporate loop information.
Thus, the calling context of a load operation looks as follows: $main()\to loop_{1}\to f()\to ...\to loop_{n} \to load_{old}$. 
Additionally, \loadspy{} maintains a 64-bit global timestamp counter $\mathcal{T}$ that is incremented when passing through each loop header and also through each load operation.
Thus, the calling context snapshot may appear as follows: $C_{old} = main()\to loop_{1}[\mathcal{T}=1]\to f()\to ...\to loop_{n}[\mathcal{T}=9]  \to load_{old}$.
We extend the calling context $E$ to be a tuple, that is, $E_{old}= \langle pointer\ to\ old\ context, \mathcal{T}_{old}\rangle$ = $\langle C_{old},10\rangle$. 

\begin{figure}[t]
\begin{minipage}[t]{.48\linewidth}
\begin{lstlisting}[firstnumber=1,language=c, caption=Redundancy in the inner loop scope., label=lst:scope1]
main () {
  // loop1
  for (i=0; i<M; i++) { 
    // loop2
    for (k=0; k<N; k++) { 
      // load from B[i] 
      t += B[i];
}}}
\end{lstlisting}
\end{minipage}\hfill
\begin{minipage}[t]{.48\linewidth}
\begin{lstlisting}[firstnumber=1,language=c, caption=Redundancy in the outer loop scope., label=lst:scope2]
main () {
  // loop1
  for (i=0; i<M; i++) { 
    // loop2
    for (k=0; k<N; k++) { 
      // load from A[k]
      t += A[k];
}}}
\end{lstlisting}
\end{minipage}
\vspace{-1.5em}
\end{figure}

Listing~\ref{lst:scope1} shows a simplified example, where the redundancy happens in the inner loop (scope is $loop_2$).
In this setting, consider the following  pair of calling context snapshot:
\begin{eqnarray*}
\scriptsize
\begin{aligned}
E_{old} =&\langle main()\to loop_{1}[\mathcal{T}=1] \to loop_{2}[\mathcal{T}=2]   \to load_{old}, \mathcal{T}_{old}=3 \rangle \\
E_{new} =&  \langle main()\to loop_{1}[\mathcal{T}=1] \to loop_{2}[\mathcal{T}=4]  \to load_{new}, \mathcal{T}_{new}=5 \rangle
\end{aligned}
\end{eqnarray*}

Notice that the counter associated with $loop_{1}$ has remained unchanged whereas the counter associated with $loop_{2}$ has changed.
Each load maintains a \emph{pointer} to the calling context, not the entire calling context snapshot.
Hence, by the time the redundancy is detected, that is, $load_{new}$ is executed, $loop_{2}[\mathcal{T}=2]$ would have gotten updated to $loop_{2}[\mathcal{T}=4]$; traversing $C_{old}$ would find $\mathcal{T}_{loop_{2}} = 4$.
Observe that $\mathcal{T}_{old} < \mathcal{T}_{loop_{2}} < \mathcal{T}_{new}$. 
This invariant informs that $loop_{2}$ is the scope inside which the redundancy is happening.
The same invariant does not hold for $\mathcal{T}_{loop_{1}}$.

Now, consider a simplified example in Listing~\ref{lst:scope2}, where redundancy happens in the outer loop (scope is $loop_{1}$).
In this setting, consider the following  pair of calling context snapshot:
\begin{eqnarray*}
\scriptsize
\begin{aligned}
E_{old} =& \langle main()\to loop_{1}[\mathcal{T}=1] \to loop_{2}[\mathcal{T}=2]  \to load_{old}, \mathcal{T}_{old}=3 \rangle \\
E_{new} =&  \langle main()\to loop_{1}[\mathcal{T}=8] \to loop_{2}[\mathcal{T}=9]  \to load_{new}, \mathcal{T}_{new}=10 \rangle
\end{aligned}
\end{eqnarray*}

Notice that the counter associated with both $loop_{1}$ and $loop_{2}$ have changed.
Hence, by the time $load_{new}$ is executed, $loop_{1}[\mathcal{T}=1]$ and  $loop_{2}[\mathcal{T}=2]$ would have gotten updated to $loop_{1}[\mathcal{T}=8]$ and $loop_{2}[\mathcal{T}=9]$, respectively;  traversing $C_{old}$ would find $\mathcal{T}_{loop_{1}} = 8$ and  $\mathcal{T}_{loop_{2}} = 9$.
Observe that $\mathcal{T}_{old} < \mathcal{T}_{loop_{1}} < \mathcal{T}_{loop_{2}} < \mathcal{T}_{new}$.
The loop with the smallest $\mathcal{T}$ value obeying this invariant, that is $loop_{1}$, is the redundancy scope.
If there was another enclosing loop, say $loop_{0}$, its counter would not have obeyed this invariant.

\begin{claim}
\label{claim:scope}
Given a redundancy context pair $\langle \langle C, \mathcal{T}_{old}\rangle , \langle  C, \mathcal{T}_{new} \rangle \rangle$,
the redundancy scope $\mathcal{S}$ is the outermost enclosing loop $i$ in $C$ such that $\mathcal{T}_{old} < \mathcal{T}_{loop_i}  < \mathcal{T}_{new}$.
\end{claim}

\begin{proof}
First, $\mathcal{T}_{loop_{i}}$ must be in the range of $(\mathcal{T}_{old}, \mathcal{T}_{new})$ because loop $i$ is the redundancy scope; otherwise, loop $i$ cannot enclose the redundant load instances. 
Next, assume there exists another loop $j$ in $C$ such that $\mathcal{T}_{old} < \mathcal{T}_{loop_{j}} < \mathcal{T}_{loop_i}<\mathcal{T}_{new}$ but loop $j$ is not the redundancy scope. 
Loop $i$ and $j$ cannot be the peer loops because they are both in the same context $C$. 
Then one loop must enclose  the other.
(1) If loop $i$ encloses loop $j$, $\mathcal{T}_{loop_i} < \mathcal{T}_{loop_{j}}$ because  loop $j$'s counter is incremented at least once after loop $i$'s counter is incremented, which contradicts the assumption that $\mathcal{T}_{loop_{j}} < \mathcal{T}_{loop_i}$. Hence, loop $j$ cannot be nested inside loop $i$.
(2) If loop $j$ encloses loop $i$, then loop $i$ is no longer the outermost loop with  $\mathcal{T}_{old} < \mathcal{T}_{loop_i}  < \mathcal{T}_{new}$. Hence, loop $j$ cannot be enclosing loop $i$.
Since loop $i$ and loop $j$ are neither peer loops, nor can they be nested within one another, the assumption is void.
Thus, Claim~\ref{claim:scope} holds.
\end{proof}

\textbf{\textit{Implementing Redundancy Scope:}}
\loadspy{} combines static and dynamic analysis to compute the redundancy scope $\mathcal{S}$ for each redundancy pair. 
First, \loadspy{} instruments each loop header in the binary (in addition to procedures) to produce calling contexts with augmented loop information.
It identifies an instruction as a loop header by performing an interval analysis~\cite{Havlak:1997:irreducible} on the binary code and integrates the information into the procedure call path.
We refer to the calling contexts with loop information as \emph{extended contexts}. 
A runtime analysis routine run as a part of each loop header increments the 64-bit timestamp counter $\mathcal{T}$.
The analysis routine run as a part of each load instruction also increments the counter $\mathcal{T}$.
Also, the shadow memory for each byte of the original program is extended to hold the counter $\mathcal{T}$ (in addition to the 32-bit calling context handle and the 8-bit old value).

On each detected load redundancy, where $C_{old} = C_{new}$, \loadspy{} searches the call path from root (\texttt{main)} toward the leaf (the \texttt{load} instruction) to look for the first loop node where the Claim~\ref{claim:scope} is found to be true.
Such a loop is the redundancy scope $\mathcal{S}$ for the current instance of load redundancy.
Each redundancy instance records the triplet $\langle C_{old}, C_{new}, \mathcal{S}\rangle$.
If $C_{old} \ne C_{new}$, \loadspy{} first finds the lowest common ancestor (LCA) function or loop enclosing $C_{old}$ and $C_{new}$, and then searches their common call path from root (main) toward the LCA to obtain $\mathcal{S}$ based on the Claim~\ref{claim:scope}.

Computing the redundancy scope for each redundancy instance introduces heavy runtime overhead. 
We compute the redundancy scope for a given calling context pair only a threshold number of times (one in our experiments), which is good enough for most programs.

\subsection{Handling Threaded Programs} 
\loadspy{} maintains per-thread data structures: calling context trees, redundancy profiles, $\mathcal{T}$, among others and hence needs no concurrency control for multi-threaded programs.
The runtime object map is maintained as a lock-free map allowing concurrent lookups.
\loadspy{} detects only intra-thread redundancy and ignores inter-thread redundancy, if any. 

\subsection{Reducing Profiling Overhead} 
 \loadspy{} can introduce relatively high runtime overhead, $\sim$40-150$\times$.
 \loadspy{} adopts a bursty sampling mechanism to control its overhead~\cite{Zhong:2008:SPL:1375634.1375648}. 
Bursty sampling involves continuous monitoring for a certain number of instructions (\texttt{WINDOW\_ENABLE}) followed by not monitoring for a certain (larger) number of instructions (\texttt{WINDOW\_DISABLE}) and repeating it over time. 
These two thresholds are tunable. 
From our experiments, 1\% sampling rate with \texttt{WINDOW\_ENABLE}=1 million and \texttt{WINDOW\_DISABLE}=99 million yields a good tradeoff between overhead and accuracy.

\subsection{Discussions} 
It is worth noting that there is no one-one relationship between the redundancy fraction and potential performance gains because of pipelining, caching and prefetching in hardware.
\loadspy{} does not distinguish actionable vs. non-actionable redundancies, which is a topic of our future work.

\section{\loadspy{} Workflow}
\label{sec:implementation}
 
\loadspy{} consists of three components: a runtime profiler (detailed previously in \S~\ref{sec:methodology}), an analyzer, and a GUI. 
\loadspy{} accepts fully optimized binary executables and collects runtime profiles via its online profiler. 
The analyzer and GUI, run in a postmortem fashion, consume the runtime profiles and 
associate them with the application source code.
The rest of this section discusses the analyzer and GUI.

\subsection{\loadspy{}'s Analyzer}
\label{subsec:analyzer}

\loadspy{}'s analyzer associates the runtime profiles with source code based on the DWARF~\cite{dwarf} information produced by compilers.
As the profiler produces per-thread profiles, the analyzer needs to coalesce the profiles for the whole execution. 
The calling context profiles scale the analysis of program execution to a large number of cores. 
The coalescing procedure follows the rule: two redundancy pairs from different threads are merged $iff$ they have the same redundant loads in the same contexts with the same redundancy scope.
All the metrics are also merged to compute unified ones across threads. 
The scheme is similar for profiles from different processes.

It is worth noting that the profile coalescing overhead grows linearly with the number of threads and processes used by the monitored program. \loadspy{} leverages the reduction tree technique~\cite{Tallent:2010:SIL:1884643.1884683} to parallelize the merging process. 
Typically, \loadspy{} takes less than one minute to produce the aggregate profiles in all of our case studies.

\subsection{\loadspy{}'s GUI}
\label{subsec:gui}

\loadspy{}'s GUI inherits the design of an existing Java-based graphical interface~\cite{adhianto2010hpctoolkit}, which enables navigating the calling contexts and the corresponding source code ordered by the monitored metrics. A top-down view shows a call path $C$ starting from \texttt{main} to a leaf function with the breakdown of metrics at each level. Merely attributing a metric to two independent contexts loses the association between two related contexts during postmortem inspection. To correlate the source with the target, \loadspy{} allows appending a copy of the target calling context to the source calling context. For example, if a load in context \texttt{main->A->B} is redundant with another load in context \texttt{main->C->D}, \loadspy{} constructs a synthetic calling context: \texttt{main->A->B->main->C->D}. The redundancy metrics will be attributed to the leaf of this call chain. These synthetic call chains make it easy to visually navigate profiles and focus on top redundancy pairs. Figure~\ref{fig:avro} in \S~\ref{subsec:avro} shows an example of the GUI, and we postpone the explanation of the GUI details to that section.

\begin{figure*}[t]
\begin{center}
\begin{subfigure}[b]{0.45\textwidth}
\includegraphics[width=.95\textwidth]{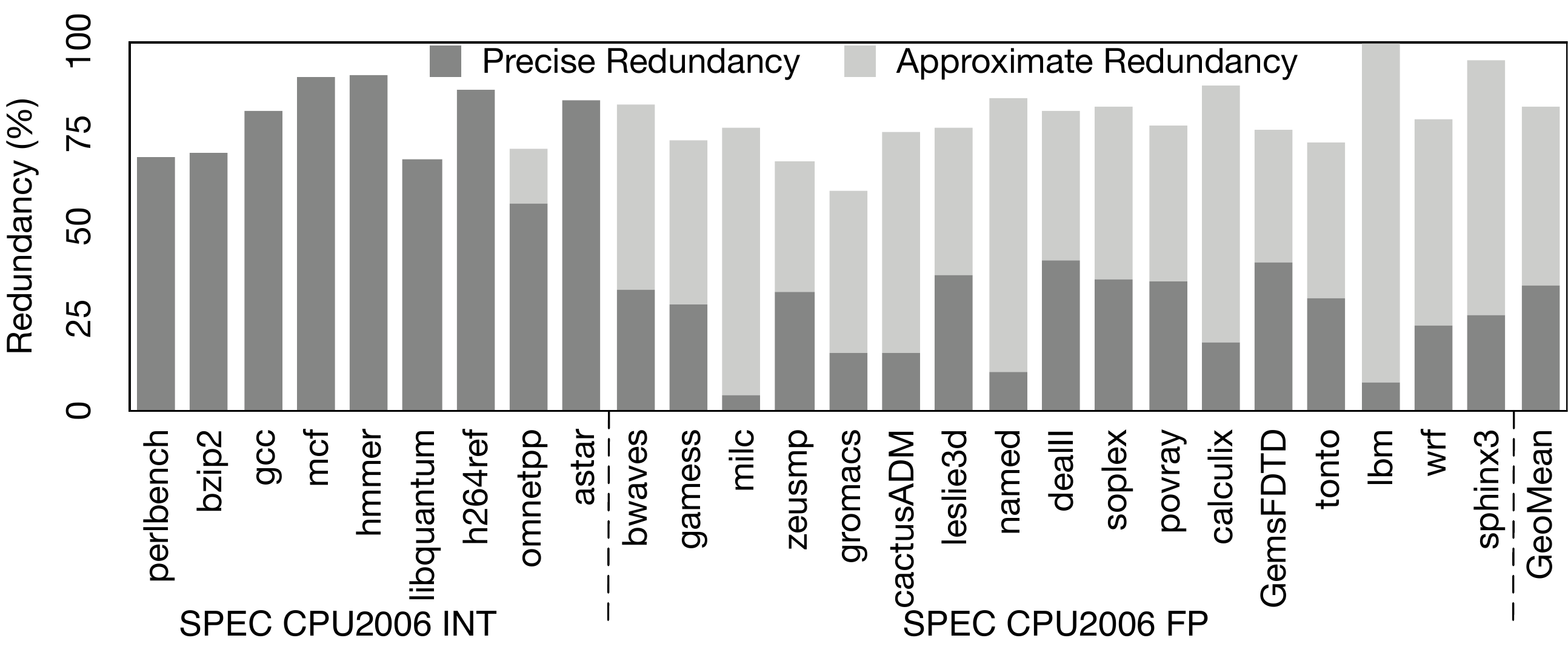}
\caption{Temporal redundancies.}
\label{fig:temporal_breakdown}
\end{subfigure}
~~
\begin{subfigure}[b]{0.45\textwidth}
\includegraphics[width=.95\textwidth]{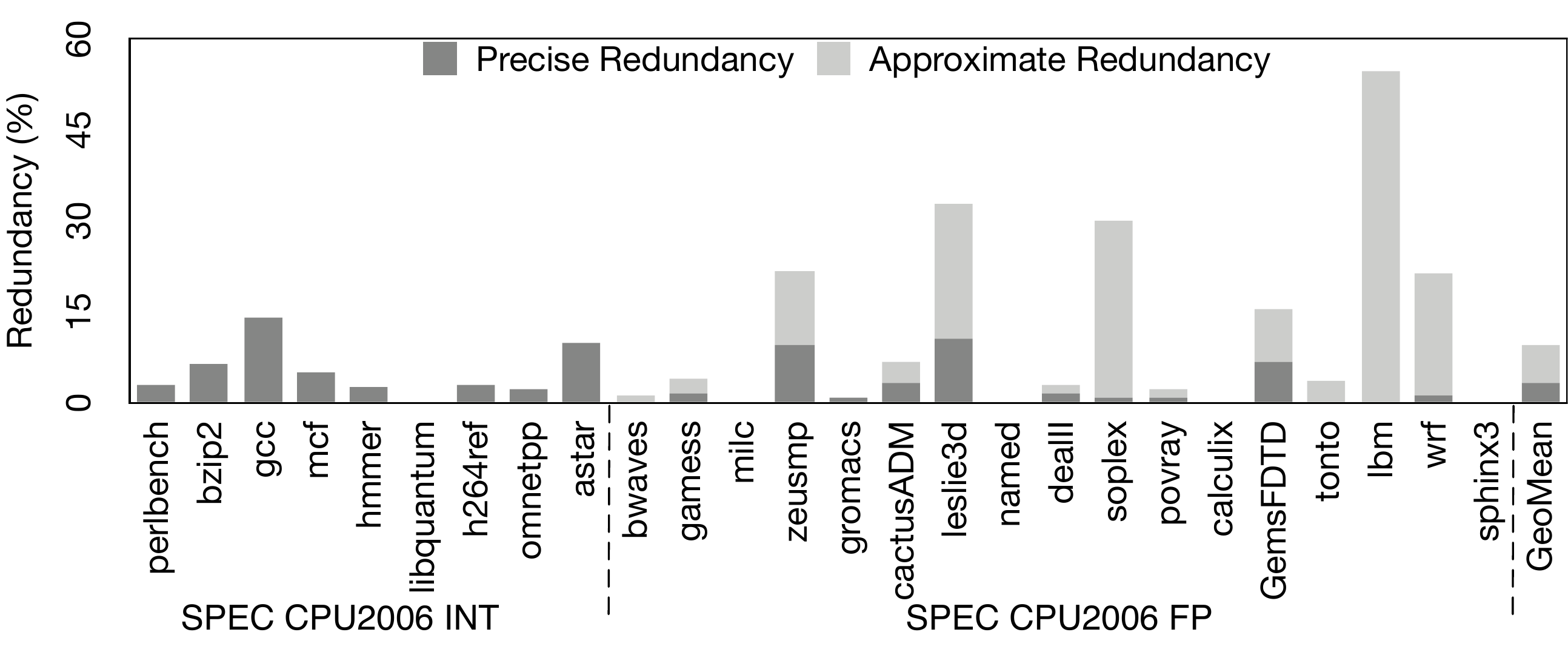}
\caption{Spatial redundancies.}
\label{fig:spatial_breakdown}
\end{subfigure}
\end{center}
\vspace{-0.1in}
\caption{Fraction of temporal and spatial load redundancies on SPEC CPU2006.}
\vspace{-.8em}
\label{fig:breakdown}
\end{figure*}

\section{Evaluation}
\label{sec:experiment}

We evaluate \loadspy{} on a 12-core Intel Xeon E5-2650 v4 CPU (Broadwell) of 2.20GHz frequency running Linux 4.8.0. The machine has 256GB main memory. We evaluate \loadspy{} with well-known benchmarks, such as SPEC CPU2006~\cite{SPEC:CPU2006}, SPEC OMP2012~\cite{SPEC:OMP2012}, SPEC CPU2017~\cite{SPEC:CPU2017}, Parsec-2.1~\cite{parsec}, Rodinia-3.1~\cite{rodinia}, NERSC-8~\cite{TRINITY-WWW}, and Stamp-0.9.10~\cite{4636089}, as well as several real-world applications, such as Apache Avro-1.8.2~\cite{Apache-Avro}, Hoard-3.12~\cite{Berger:2000:HSM:378993.379232}, MASNUM-2.2~\cite{Qiao:2016:HEG:3014904.3014911}, Shogun-6.0~\cite{soeren_sonnenburg_2017_556748}, USQCD Chroma-3.43~\cite{Edwards:2004sx}, Stack RNN~\cite{2015arXiv150301007J}, Binutils-2.27~\cite{binutils}, and Kallisto-0.43~\cite{kallisto-WWW}. 
All the programs are compiled with \texttt{gcc-4.8.5 -O3 PGO} except Hoard-3.12 and MASNUM-2.2. For Hoard-3.12 we use \texttt{clang-5.0.0 -O3  PGO} and for MASNUM-2.2 we use \texttt{icc-17.0.4 -O3 PGO}. We apply the \texttt{ref} inputs for SPEC CPU2006, OMP2012 and CPU2017 benchmarks, the native inputs for Parsec-2.1 benchmarks, and the default inputs released with the remaining benchmarks and applications if not specified. We run all the parallel programs with four threads with simultaneous multi-threading (SMT) disabled. 

In the rest of this section, we first show the fraction of temporal and spatial redundancies obtained from SPEC CPU2006. We then evaluate the accuracy and overhead of \loadspy{} with bursty sampling enabled. We exclude three benchmarks---gobmk, sjeng, and xalancbmk---from monitoring because they have deep call recursion causing \loadspy{} to run out of memory.

\paragraph{\textbf{Load redundancy in macro benchmarks}}
Figure~\ref{fig:breakdown} shows the fraction of temporal and spatial load redundancies on SPEC CPU2006.
We can see (1) load redundancy, especially the temporal one, pervasively exists
and (2) integer benchmarks show a high proportion of precise redundant loads whereas floating-point benchmarks show a high proportion of approximate redundant loads, as expected.

\begin{table}[t]
\centering
\scriptsize
\centering 
\begin{adjustbox}{width=0.43\textwidth}
\begin{tabular}  {|c||c|c||c|c|} 
\hline 
\multirow{2}{*}{Benchmarks} & \multicolumn{2}{c||}{Detecting Temporal Redundancy} & \multicolumn{2}{c|}{Detecting Spatial Redundancy}\\ \cline{2-5}
& Runtime Slowdown & Memory Bloat  & Runtime Slowdown & Memory Bloat \\
\hline\hline 
perlbench & 38$\times$ & 11$\times$ & 51$\times$ & 7$\times$ \\
bzip2 & 13$\times$ & 2$\times$ & 13$\times$ & 1.09$\times$ \\
gcc & 19$\times$ & 26$\times$ & 19$\times$ & 25$\times$ \\
mcf & 6$\times$ & 14$\times$ & 6$\times$ & 1.04$\times$  \\
hmmer & 12$\times$ & 35$\times$ & 11$\times$ & 20$\times$ \\
libquantum & 12$\times$ & 18$\times$ & 13$\times$ & 2$\times$ \\
h264ref & 21$\times$ & 20$\times$ & 21$\times$ & 2$\times$ \\
omnetpp & 10$\times$ &16$\times$ & 14$\times$ & 25$\times$ \\
astar & 11$\times$ & 13$\times$ & 11$\times$ & 18$\times$ \\
bwaves & 17$\times$ & 14$\times$ & 15$\times$ & 1.16$\times$ \\
gamess & 24$\times$ & 25$\times$ & 24$\times$ & 24$\times$ \\
milc & 4$\times$ & 10$\times$ & 4$\times$ & 1.18$\times$ \\
zeusmp & 8$\times$ & 14$\times$ & 7$\times$ & 1.42$\times$ \\
gromacs & 10$\times$ & 23$\times$ & 9$\times$ & 15$\times$ \\
cactusADM & 7$\times$ & 10$\times$ & 7$\times$ & 1.36$\times$ \\
leslie3d & 9$\times$ & 10$\times$ & 8$\times$ & 2$\times$ \\
named & 10$\times$ & 11$\times$ & 10$\times$ & 9$\times$ \\
dealII & 21$\times$ & 30$\times$ & 22$\times$ & 19$\times$ \\
soplex & 13$\times$ & 13$\times$ & 13$\times$ & 2$\times$ \\
povray & 29$\times$  & 216$\times$ & 28$\times$ & 70$\times$ \\
calculix & 21$\times$ & 18$\times$ & 20$\times$ & 19$\times$ \\
GemsFDTD & 8$\times$ & 14$\times$ & 8$\times$ & 1.42$\times$ \\
tonto & 22$\times$ & 49$\times$ & 24$\times$ & 30$\times$ \\
lbm & 4$\times$ & 14$\times$ & 3$\times$ & 1.15$\times$ \\
wrf & 15$\times$ & 10$\times$ & 16$\times$ & 3$\times$ \\
sphinx3 & 13$\times$ & 16$\times$ & 13$\times$ & 7$\times$ \\  \hline
\textbf{Median}  & \textbf{12.5$\times$} & \textbf{14$\times$} & \textbf{13$\times$} & \textbf{5$\times$} \\ 
\textbf{GeoMean}  & \textbf{13$\times$} & \textbf{17$\times$} & \textbf{13$\times$} & \textbf{5$\times$} \\ \hline
\end{tabular}
\end{adjustbox}
\caption{\loadspy{}'s runtime slowdown and memory bloat over native execution on SPEC CPU2006.}
\vspace{-1.5em}
\label{tab:overhead}
\end{table}

\begin{figure*}[t]
\begin{center}
\begin{subfigure}[b]{0.49\textwidth}
\includegraphics[width=\textwidth]{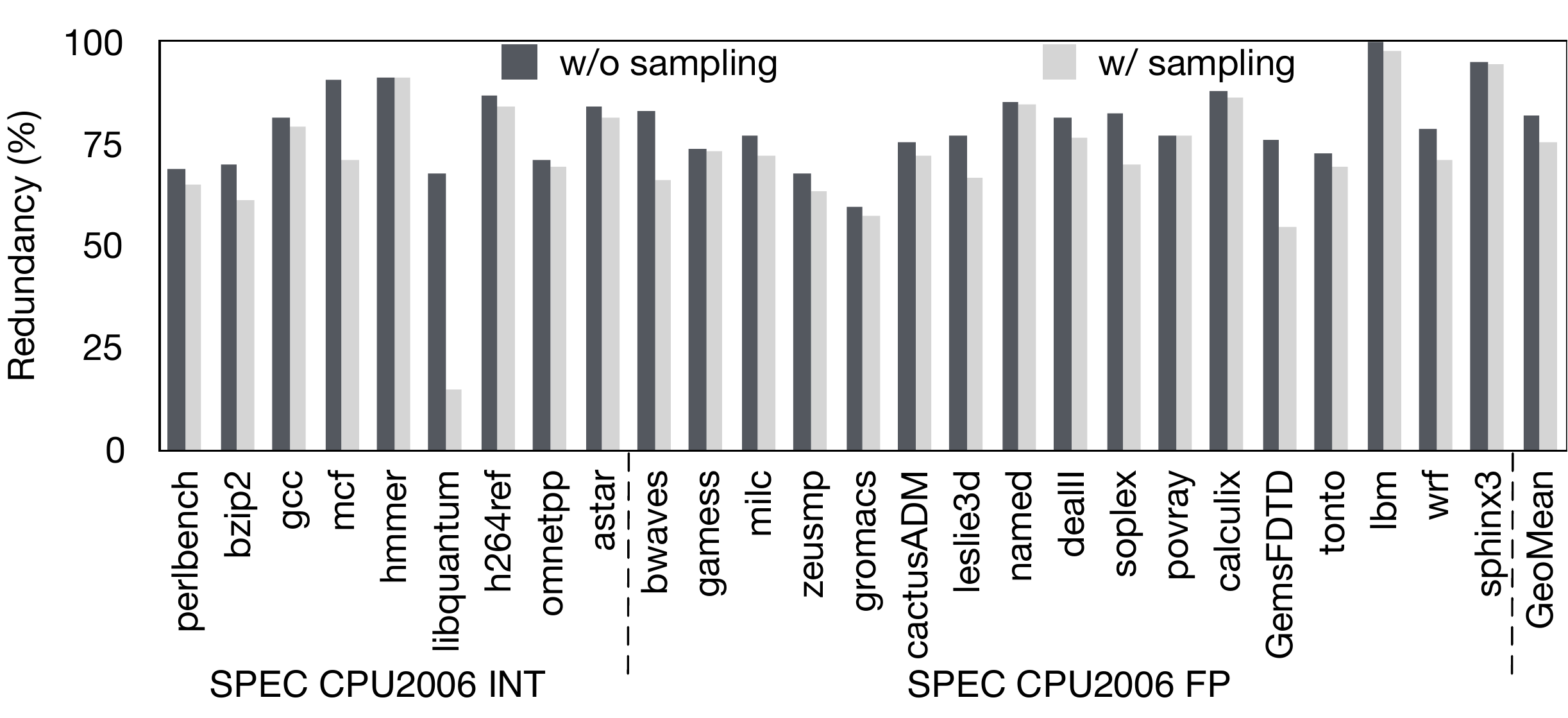}
\caption{Temporal redundancies.}
\end{subfigure}
~
\begin{subfigure}[b]{0.49\textwidth}
\includegraphics[width=\textwidth]{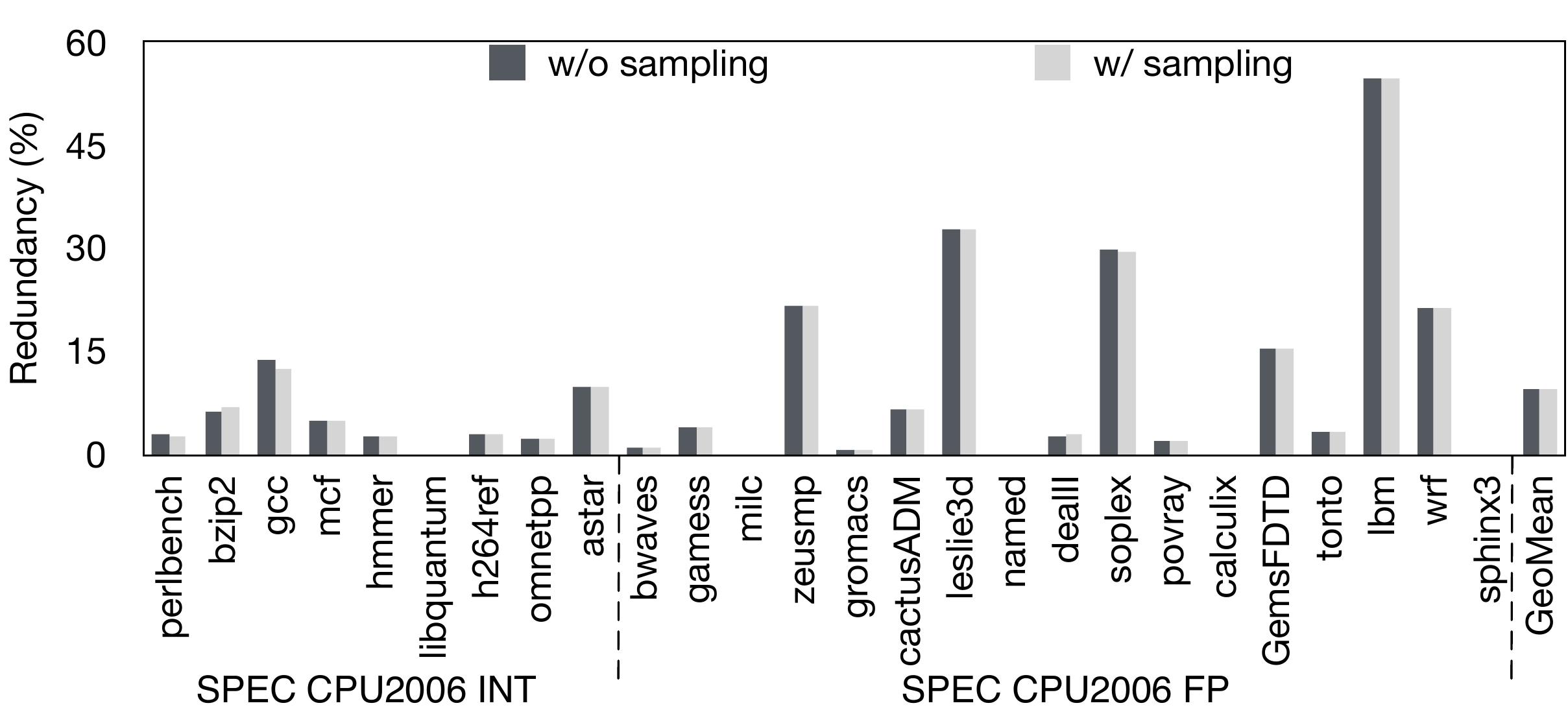}
\caption{Spatial redundancies.}
\end{subfigure}
\end{center}
\vspace{-0.15in}
\caption{Comparing temporal and spatial load redundancies with bursty sampling disabled and enabled. The sampling rate is 1\%.}
\label{fig:accuracy}
\end{figure*}

\paragraph{\textbf{Accuracy}} 
\loadspy{} offers bursty sampling as an optional feature for users willing to tradeoff measurement accuracy with performance.
Figure~\ref{fig:accuracy} evaluates the accuracy of \loadspy{} with bursty sampling enabled.
The geo-means of spatial load redundancy fractions \loadspy{} measures with sampling enabled and disabled are nearly the same---10\%.
The geo-means of temporal load redundancy fractions \loadspy{} measures with sampling enabled and disabled are similar---76\% and 82\%.
However, \texttt{libquantum} is an outlier, whose temporal redundancy fractions are 15\% and 68\% with sampling enabled and disabled. 
With further investigation, we find that the average number of instructions executed between the source and target load operations of most redundancy pairs is more than 10 million, which is greater than the default \texttt{WINDOW\_ENABLE} (= 1 million). 
In such a case, one can enlarge \texttt{WINDOW\_ENABLE} to improve the accuracy. 
For instance, when we set \texttt{WINDOW\_ENABLE} = 10 million and 50 million (\texttt{WINDOW\_DISABLE} remains unchanged), the temporal load redundancy fraction of \texttt{libquantum} increases to 30\% and 60\%, respectively. 

\paragraph{\textbf{Overhead}}
Table~\ref{tab:overhead} shows the runtime slowdown and memory bloat of \loadspy{} on SPEC CPU2006. The runtime slowdown (memory bloat) is measured as the ratio of the runtime (peak memory usage) of a benchmark with \loadspy{} enabled to the runtime (peak memory usage) of its native execution.
The geo-means of runtime slowdown for detecting temporal and spatial redundancies are both 13$\times$, and
the geo-means of memory bloat for detecting temporal and spatial redundancies are 17$\times$ and 5$\times$, respectively. 
A few benchmarks such as \texttt{tonto} and \texttt{povray} show excessive memory bloat due to the following reasons:
(1) \texttt{tonto} has a deep call stack, which demands excessive space to maintain its calling context tree and
(2) \texttt{povray} has a small ($\sim$6MB) memory footprint, whereas some preallocated data structures in \loadspy{} overshadow this baseline memory footprint.

\begin{table*}
\centering 
\begin{adjustbox}{width=0.88\textwidth}
\scriptsize
\begin{tabular}{|c|c|c||c|c||c|c|} 
\hline
\multicolumn{3}{|c||}{Program Information} & \multicolumn{2}{c||} {\loadspy{}} & \multicolumn{2}{c|} {Optimization} \\ \cline{1-7}
\multicolumn{2}{|c|}{Programs} & Problematic Code & Redundancy Types & Inefficiencies & Approaches & WS$^*$ \\  \hline
\hline 
\multirow{15}{*}{\rot{Macro Benchmarks}} & 359.botsspar & sparselu.c:loop(191) & Temporal & Inefficient register usage & Scalar replacement & 1.77$\times$\\ \cline{2-7}
& 453.povray & csg.cpp(250) & Temporal & Missing inline substitution & Function inlining & 1.05$\times$\\
& 464.h264ref & mv-search.c:loop(394) & Temporal & Missing inline substitution & Function inlining & 1.28$\times$\\
& \cmark 470.lbm & lbm.c:LBM\_performStreamCollide & Spatial & Redundant computation & Approximate computing & 1.25$\times$\\ \cline{2-7}
& \cmark 538.imagick\_r & morphology.c:loop(2982) & Spatial & Redundant computation & Conditional computation & 1.25$\times$\\ \cline{2-7}
& \cmark backprop & backprop.c:loop(322) & Spatial& Input-sensitive redundancy & Conditional computation & 1.13$\times$\\
& \cmark hotspot3D & 3D.c:loop(98, 166) & Temporal&  Inefficient register usage & Scalar replacement  & 1.13$\times$\\
& \cmark lavaMD & kernel\_cpu.c(175) & Temporal & Redundant function calls & Reusing the previous result & 1.39$\times$\\
& \cmark srad\_v1 & main.c:loop(256) & Temporal& Inefficient register usage  & Scalar replacement  & 1.11$\times$\\
& \cmark srad\_v2 & srad.cpp:loop(131) & Temporal& Inefficient register usage  & Scalar replacement  & 1.12$\times$\\ 
& \cmark particlefilter & ex\_particle\_OPENMP\_seq.c:findIndex & Temporal & Linear search & Binary search & 9.8$\times$\\ \cline{2-7}
& vacation & client.c:loop(198) & Temporal & Redundant function calls & Reusing the previous result & 1.23$\times$\\ \cline{2-7}
& dedup & hashtable.c:hashtable\_search & Temporal & Poor hashing &  Reducing hash collisions & 1.11$\times$\\ \cline{2-7}
& msgrate & msgrate.c:cache\_invalidate & Temporal & Missing constant propagation & Copy propagation & 3.03$\times$\\  \hline
\multirow{10}{*}{\rot{Real Applications}} & \cmark Apache Avro-1.8.2 & Specific.hh(110, 117) & Temporal & Missing inline substitution & Function inlining & 1.19$\times$\\ \cline{2-7}
& \cmark Hoard-3.12 & libhoard.cpp:xxmalloc & Temporal & Redundant computation & Reusing the previous result & 1.14$\times$\\ \cline{2-7}
& \cmark MASNUM-2.2 & propagat.inc:loop(130, 140) & Temporal& Linear search & Locality-friendly search & 1.79$\times$\\ \cline{2-7}
& \cmark USQCD Chroma-3.43 & qdp\_random.h(56) & Temporal & Missing inline substitution & Function inlining & 1.06$\times$\\ \cline{2-7}
& \cmark Shogun-6.0 & \makecell{DenseFeatures.cpp(505) \\Distance.cpp(185)} & Temporal& Missing inline substitution & Function inlining & 1.06$\times$\\ \cline{2-7}
& \cmark Stack RNN & StackRNN.h:loop(350, 355, 363, 367) & \makecell{Temporal \\Spatial} & \makecell{Poor choice of algorithm \\Redundant computation} & \makecell{Loop fusion \\Conditional computation} & 1.09$\times$\\ \cline{2-7}
& Kallisto-0.43 & KmerHashTable.h(131) & Temporal & Poor hashing & Reducing hash collisions & 4.1$\times$\\ \cline{2-7}
& Binutils-2.27 & dwarf2.c:loop(2166) & Temporal & Linear search & Binary search & 3.29$\times$\\ \hline

\multicolumn{4}{l}{{\vbox to 2ex{\vfil}}\scriptsize \cmark: newfound performance bugs via \loadspy{}. } \\
\multicolumn{4}{l}{{\vbox to 2ex{\vfil}}\scriptsize WS$^*$: whole-program speedup after problem elimination.} \\
\end{tabular}
\end{adjustbox}
\caption{Overview of performance improvement guided by \loadspy{}.}
\label{tab:perf}
\end{table*}

\section{Case Studies}
\label{sec:use}

We evaluate the load redundancies found in some benchmarks and real-world applications. 
Table~\ref{tab:perf} summarizes the inefficiencies found and the speedups obtained by eliminating them. 
We quantify the performance of all programs in execution time except Hoard in throughput. 
In the rest of this section, we exhaustively analyze all the newfound performance bugs. 

\subsection{Apache Avro-1.8.2}
\label{subsec:avro}
 
Avro~\cite{Apache-Avro} is a remote procedure call (RPC) and data serialization processing system. We apply \loadspy{} to evaluate the C++ version of Avro with benchmarks developed by Sorokin~\cite{avro-benchmark-WWW}. \loadspy{} reports a temporal redundancy fraction $\mathcal{R}_{prog}^{precise}$ of 79\% for the entire program.
Figure~\ref{fig:avro} shows the full calling contexts of the top redundancy pair visualized through \loadspy{}'s GUI. \loadspy{}'s GUI consists of three panes: the top pane shows the program source code, the bottom left pane shows the full calling contexts of each redundancy pair, and the bottom right pane shows the metrics associated with each redundancy pair. In this figure, the GUI shows two metrics:  the number of redundant loads for a given redundancy pair and percentage of redundant instances for a given pair, which if 100\%, means every instance of this pair is redundant.

From the figure, we can see that the redundant loads in function \texttt{doEncodeLong} account for 25\% of the total redundant loads in the program. 
Moreover, all instances of this pair are redundant. The redundancy scope of this pair is the loop at lines 229-233 in the file \texttt{Specific.hh} enclosing the call site of function \texttt{encode}. Function \texttt{encode} is the caller of function \texttt{doEncodeLong}.
With further analysis, we find that the epilog of function \texttt{doEncodeLong} consistently pops the same values from the same stack location to restore the register values. 
To eliminate redundant loads in the function epilog, we inline \texttt{doEncodeLong} into its caller. 
\loadspy{} further identifies another problematic function (not shown) and guides the same inlining optimization. 
Together, these optimizations eliminate 31\% of the memory loads and 37\% of the redundant memory loads, yielding a 1.19$\times$ speedup for the whole program.
 
\begin{figure}
\begin{center}
\includegraphics[width=0.49\textwidth]{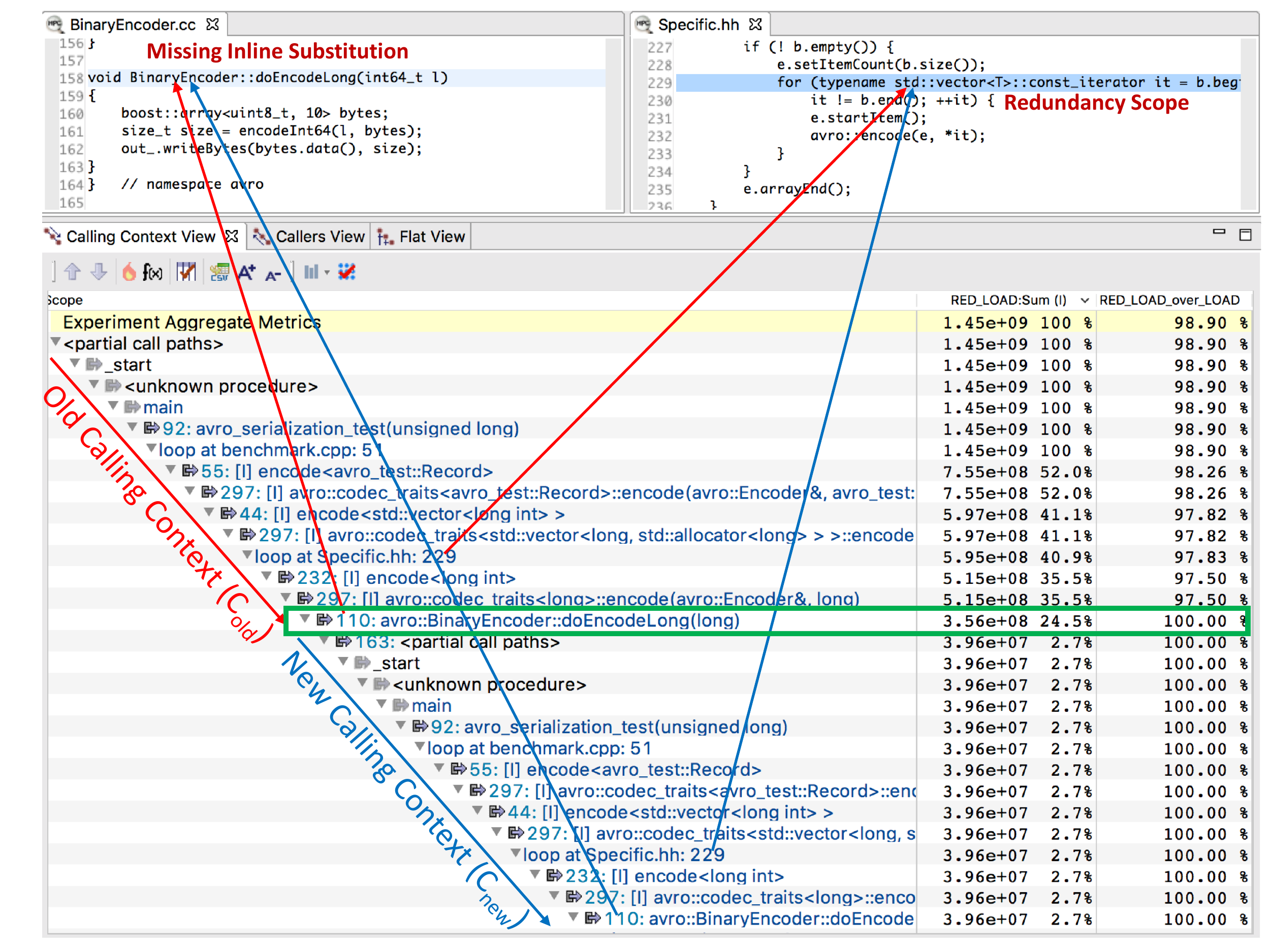}
\end{center}
\vspace{-0.15in}
\caption{The top redundancy pair in \texttt{Avro} with full calling contexts reported by \loadspy{}. Along the calling contexts shown in the bottom left pane, a procedure name following a symbol \texttt{[I]} means it is inlined. We can see that most procedures on the path are inlined, except \texttt{doEncodeLong}. 
Many redundant loads are from calling \texttt{doEncodeLong}, which can be removed by function inlining.}
\vspace{-1.2em}
\label{fig:avro}
\end{figure}

\subsection{MASNUM-2.2}
\label{subsec:masnum}
  
MASNUM~\cite{Qiao:2016:HEG:3014904.3014911}, one of the 2016 ACM Gordon Bell Prize finalists, forecasts ocean surface waves and climate change. It is written in Fortran and parallelized with MPI. \loadspy{} identifies 91\% of memory loads are redundant, of which 15\% are attributed to the array \texttt{x} at line 6 on the left of Listing~\ref{lst:motivationExample}. \loadspy{} also pinpoints the redundancy scope as the outermost loop at line 1. We find that the innermost loop (line 5) performs a linear search over the non-decreasing array \texttt{x} for a given input \texttt{xx}. With multiple iterations, elements of array \texttt{x} are frequently loaded from memory for comparison, leading to the redundancy. 
Changing the linear search to a binary search reduces redundant loads and yields a 1.32$\times$ speedup for the entire program. 
It is worth noting that the binary search still incurs high load redundancy fraction because of the intensive search requests in the program.
To further improve the search algorithm, we analyze the values of \texttt{xx} across iterations. 
We find that \texttt{xx} has good value locality, that is, the values are similar in adjacent iterations of the outermost loop.
Thus, we replace the binary search with a locality-friendly search. We memoize the location index \texttt{iii} when the current search finishes; in the next search, we begin at the recorded \texttt{iii} and alternate the linear search in both directions to the array start and end.
This optimization eliminates 33\% of the memory loads and 36\% of the redundant memory loads, yielding a 1.79$\times$ speedup for the entire program.
 
\subsection{Hoard-3.12}
\label{subsec:hoard}
   
Hoard~\cite{Berger:2000:HSM:378993.379232}, a high-performance cross-platform C++ based memory allocator, has been integrated into an array of applications and programming languages such as GNU Bayonne and Cilk programming language. It has 20K lines of code and is parallelized with the \texttt{PThreads} library. 
\loadspy{} identifies that 58\% of memory loads are redundant on profiling Hoard's built-in benchmark \texttt{larson}. The top redundancy pair is associated with lines 4 and 7 shown in Listing~\ref{lst:hoard}, which accounts for 11\% of the total redundant loads. 
The cause of such redundancy is that the program repeatedly checks whether \texttt{theTLAB} is a null pointer. More specifically, function \texttt{isCustomHeapInitialized} at line 15 and function \texttt{getCustomHeap} at line 16 both include code to check whether \texttt{theTLAB} is equal to \texttt{nullptr}. Hence, the second check at lines 8-11 in \texttt{getCustomHeap} is redundant. 
 
To eliminate such redundant loads, we inline these two functions into their caller \texttt{xxmalloc} and remove the redundant check.
This optimization eliminates 3\% of the memory loads and 2\% of the redundant memory loads, which improves the throughput (i.e., the number of memory operations per second) of Hoard by 1.14$\times$.

\begin{figure}[t] 
\begin{lstlisting}[firstnumber=1,language=c, caption=Temporal load redundancy in Hoard-3.12. The program repeatedly checks whether the pointer variable \texttt{theTLAB} is null., label=lst:hoard]
static __thread TheCustomHeapType * theTLAB INITIAL_EXEC_ATTR = nullptr;
...
bool isCustomHeapInitialized() {
@$\blacktriangleright$@ return (theTLAB != nullptr);
}
TheCustomHeapType * getCustomHeap() {
@$\blacktriangleright$@ auto tlab = theTLAB;
  if (tlab == nullptr) {
    tlab = initializeCustomHeap();
    theTLAB = tlab;
  }
  return tlab;
}
void * xxmalloc (size_t sz) {
  if (isCustomHeapInitialized()) {
   void * ptr = getCustomHeap()->malloc(sz);
   ... 
  }
}
\end{lstlisting}
\end{figure}

\subsection{USQCD Chroma-3.43}
\label{subsec:chroma}

Chroma~\cite{Edwards:2004sx} is a complex toolbox for performing quantum chromodynamics lattice computations, 
which has more than 200K lines of code. We evaluate it using the built-in benchmark \texttt{t\_mesplq}. 
\loadspy{} reports a temporal redundancy fraction of 61\%. The top redundancy pair is attributed to the function \texttt{sranf} at line 3 shown in Listing~\ref{lst:chroma}. With further investigation, we notice that Chroma has a similar performance bug to the one in \texttt{Apache Avro}: the epilog of function \texttt{sranf} repeatedly pops the same values from the same stack location to restore the register values. 

To eliminate such redundant loads, we manually inline the callee into its caller. 
This optimization eliminates 6\% of the memory loads and 7\% of the redundant memory loads, yielding a 1.06$\times$ speedup for the whole program.

\begin{figure}[t] 
\begin{lstlisting}[firstnumber=1,language=fortran, caption=Temporal load redundancy in USQCD Chroma-3.43. The epilog of function \texttt{sranf} often pops the same values from the same stack location to restore the register values., label=lst:chroma]
template<class T1, class T2>
inline void fill_random(float& d, T1& seed, T2& skewed_seed, const T1& seed_mult) {
@$\blacktriangleright$@ d = float(RNG::sranf(seed, skewed_seed, seed_mult));
}
\end{lstlisting}
\end{figure}

\subsection{Shogun-6.0}
\label{subsec:shogun}

Shogun~\cite{soeren_sonnenburg_2017_556748} is an efficient machine learning toolbox. \loadspy{} reports a temporal redundancy fraction of 71\% on profiling its built-in benchmark \texttt{kernel\_matrix\_sum\_benchmark}. Listing~\ref{lst:shogun} shows one of the top redundancy pairs at line 6. The cause of such redundancy is similar to \texttt{Apache Avro}: the epilog of function \texttt{get\_feature\_vector} repeatedly pops the same values from the same stack location to restore the register values. We manually inline the callee into its caller to eliminate these redundant loads. Additionally, We perform the same optimization for other function invocations that have the same performance issue. These optimizations eliminate 7\% of the memory loads and 2\% of the redundant memory loads, yielding a $1.06\times$ speedup for the whole program.

\begin{figure}[t]
\begin{lstlisting}[firstnumber=1,language=c, caption= {Temporal load redundancy in Shogun-6.0. The epilog of function \texttt{get\_feature\_vector} often pops the same values from the same stack location to restore the register values.}, label=lst:shogun]
template<class ST> float64_t CDenseFeatures<ST>::dot(int32_t vec_idx1, CDotFeatures* df, int32_t vec_idx2) {
  ...
  CDenseFeatures<ST>* sf = (CDenseFeatures<ST>*) df;
  int32_t len1, len2;
  bool free1, free2;
@$\blacktriangleright$@ ST* vec1 = get_feature_vector(vec_idx1, len1, free1);
  ...
}
\end{lstlisting}
\end{figure}

\subsection{Stack RNN}
\label{subsec:stackRnn}

Stack RNN~\cite{2015arXiv150301007J} is a C++ based project originating from Facebook AI research, which applies memory stack to optimize and extend a recurrent neural network. 
We evaluate Stack RNN by profiling its built-in application \texttt{train\_add} with \loadspy{}. \loadspy{} quantifies a redundancy fraction of 81\%, and pinpoints that the top temporal and spatial load redundancy pairs are associated with four loops shown in Listing~\ref{lst:stackRnn}. 

The cause of the temporal load redundancy is that each of the four loops accesses array \texttt{\_err\_stack}. However, the compiler cannot keep all elements of array \texttt{\_err\_stack} in CPU registers across these loops. 
Thus, the elements of array \texttt{\_err\_stack} are repeatedly loaded from memory into registers. 
We eliminate the temporal redundant loads by loop fusion, which fuses the four loops into one so that array \texttt{\_err\_stack} is only loaded once. 
	
The cause of the spatial load redundancy is that most elements of array \texttt{\_err\_stack} are zeros, resulting in identity computation at lines 2, 5, 9 and 12 shown in Listing~\ref{lst:stackRnn}. 
We employ a conditional check to avoid the computation on identities. 
These two optimizations together eliminate 10\% of the memory loads and 15\% of the redundant memory loads, yielding a 1.09$\times$ speedup for the whole program.
	
\begin{figure}[t]
\begin{lstlisting}[firstnumber=1,language=c,  caption={Temporal and spatial load redundancies in Stack RNN. Array \texttt{\_err\_stack} is loaded from memory by each of the four loops, resulting in temporal load redundancy. Besides, most elements of array \texttt{\_err\_stack} equal zero, resulting in spatial load redundancy.}, label=lst:stackRnn]
for (my_int i = _TOP_OF_STACK; i < _TOP_OF_STACK + _STACK_SIZE - 1; i++) {
@$\blacktriangleright$@ _pred_err_stack[s][i+1] += _err_stack[s][i] * _act[s][itm][pop];
 }
for (my_int i = _TOP_OF_STACK; i < _TOP_OF_STACK + _STACK_SIZE - 1; i++) {
@$\blacktriangleright$@ _err_act[s][pop] += _err_stack[s][i] * _stack[s][old_it][i+1];
 }
_err_act[s][pop] += _err_stack[s][_TOP_OF_STACK + _STACK_SIZE - 1] * EMPTY_STACK_VALUE;
for (my_int i = _TOP_OF_STACK + 1; i < _TOP_OF_STACK + _STACK_SIZE; i++) {
@$\blacktriangleright$@ _pred_err_stack[s][i-1] += _err_stack[s][i] * _act[s][itm][push];
 }
 for (my_int i = _TOP_OF_STACK + 1; i < _TOP_OF_STACK + _STACK_SIZE; i++) {
@$\blacktriangleright$@ _err_act[s][push] += _err_stack[s][i] * _stack[s][old_it][i-1];
 }
\end{lstlisting} 
\end{figure}

\subsection{SPEC CPU2006 470.lbm}
\label{subsec:lbm}

470.lbm~\cite{SPEC:CPU2006} employs the lattice boltzmann method to simulate incompressible fluids in three-dimensional space. \loadspy{} reports that spatial redundant loads account for 55\% of the total memory loads, of which more than 30\% are attributed to the array \texttt{srcGrid} at lines 11-54 shown in Listing~\ref{lst:lbm}. With further investigation, we find that array \texttt{srcGrid} is traversed across loop iterations and most of its elements are identical, resulting in many redundant loads.

To optimize this inefficiency, we apply loop perforation~\cite{Sidiroglou-Douskos:2011:MPV:2025113.2025133} to reduce the number of iterations at the cost of accuracy. With this optimization, the memory loads and redundant memory loads are reduced by 26\% and 60\%, and the whole program gains a 1.25$\times$ with trivial accuracy loss (7.7e-5\%).

\begin{figure}[t]
\begin{lstlisting}[firstnumber=1,language=c, caption= Spatial load redundancy in  SPEC CPU2006 470.lbm. Array \texttt{srcGrid} is frequently loaded from memory while most array elements have the same values., label=lst:lbm]
#define SWEEP_START(x1,y1,z1,x2,y2,z2) \
  for( i = CALC_INDEX(x1, y1, z1, 0); \
  i < CALC_INDEX(x2, y2, z2, 0); \
  i += N_CELL_ENTRIES ) {
#define SWEEP_END }
...
static double srcGrid[SIZE_Z*SIZE_Y*SIZE_X*N_CELL_ENTRIES];
...
SWEEP_START( 0, 0, 0, 0, 0, SIZE_Z ) // loop entry
...
@$\blacktriangleright$@ rho = + SRC_C ( srcGrid ) + SRC_N ( srcGrid )
@$\blacktriangleright$@ + SRC_S ( srcGrid ) + SRC_E ( srcGrid )
@$\blacktriangleright$@ + SRC_W ( srcGrid ) + SRC_T ( srcGrid )
@$\blacktriangleright$@ + SRC_B ( srcGrid ) + SRC_NE( srcGrid )
@$\blacktriangleright$@ + SRC_NW( srcGrid ) + SRC_SE( srcGrid )
@$\blacktriangleright$@ + SRC_SW( srcGrid ) + SRC_NT( srcGrid )
@$\blacktriangleright$@ + SRC_NB( srcGrid ) + SRC_ST( srcGrid )
@$\blacktriangleright$@ + SRC_SB( srcGrid ) + SRC_ET( srcGrid )
@$\blacktriangleright$@ + SRC_EB( srcGrid ) + SRC_WT( srcGrid )
@$\blacktriangleright$@ + SRC_WB( srcGrid );
@$\blacktriangleright$@ ux = + SRC_E ( srcGrid ) - SRC_W ( srcGrid )
@$\blacktriangleright$@ + SRC_NE( srcGrid ) - SRC_NW( srcGrid )
@$\blacktriangleright$@ + SRC_SE( srcGrid ) - SRC_SW( srcGrid )
@$\blacktriangleright$@ + SRC_ET( srcGrid ) + SRC_EB( srcGrid )
@$\blacktriangleright$@ -SRC_WT( srcGrid ) - SRC_WB( srcGrid );
@$\blacktriangleright$@ uy = + SRC_N ( srcGrid ) - SRC_S ( srcGrid )
@$\blacktriangleright$@ + SRC_NE( srcGrid ) + SRC_NW( srcGrid )
@$\blacktriangleright$@ - SRC_SE( srcGrid ) - SRC_SW( srcGrid )
@$\blacktriangleright$@ + SRC_NT( srcGrid ) + SRC_NB( srcGrid )
@$\blacktriangleright$@ - SRC_ST( srcGrid ) - SRC_SB( srcGrid );
@$\blacktriangleright$@ uz = + SRC_T ( srcGrid ) - SRC_B ( srcGrid )
@$\blacktriangleright$@ + SRC_NT( srcGrid ) - SRC_NB( srcGrid )
@$\blacktriangleright$@ + SRC_ST( srcGrid ) - SRC_SB( srcGrid )
@$\blacktriangleright$@ + SRC_ET( srcGrid ) - SRC_EB( srcGrid )
@$\blacktriangleright$@ + SRC_WT( srcGrid ) - SRC_WB( srcGrid );
...
@$\blacktriangleright$@ DST_C ( dstGrid ) = (1.0-OMEGA)*SRC_C ( srcGrid ) + ...
@$\blacktriangleright$@ DST_N ( dstGrid ) = (1.0-OMEGA)*SRC_N ( srcGrid ) + ...
@$\blacktriangleright$@ DST_E ( dstGrid ) = (1.0-OMEGA)*SRC_E ( srcGrid ) + ...
@$\blacktriangleright$@ DST_W ( dstGrid ) = (1.0-OMEGA)*SRC_W ( srcGrid ) + ...
@$\blacktriangleright$@ DST_T ( dstGrid ) = (1.0-OMEGA)*SRC_T ( srcGrid ) + ...
@$\blacktriangleright$@ DST_B ( dstGrid ) = (1.0-OMEGA)*SRC_B ( srcGrid ) + ...
@$\blacktriangleright$@ DST_NE( dstGrid ) = (1.0-OMEGA)*SRC_NE( srcGrid ) + ...
@$\blacktriangleright$@ DST_NW( dstGrid ) = (1.0-OMEGA)*SRC_NW( srcGrid ) + ...
@$\blacktriangleright$@ DST_SE( dstGrid ) = (1.0-OMEGA)*SRC_SE( srcGrid ) + ...
@$\blacktriangleright$@ DST_SW( dstGrid ) = (1.0-OMEGA)*SRC_SW( srcGrid ) + ...
@$\blacktriangleright$@ DST_NT( dstGrid ) = (1.0-OMEGA)*SRC_NT( srcGrid ) + ...
@$\blacktriangleright$@ DST_NB( dstGrid ) = (1.0-OMEGA)*SRC_NB( srcGrid ) + ...
@$\blacktriangleright$@ DST_ST( dstGrid ) = (1.0-OMEGA)*SRC_ST( srcGrid ) + ...
@$\blacktriangleright$@ DST_SB( dstGrid ) = (1.0-OMEGA)*SRC_SB( srcGrid ) + ...
@$\blacktriangleright$@ DST_ET( dstGrid ) = (1.0-OMEGA)*SRC_ET( srcGrid ) + ...
@$\blacktriangleright$@ DST_EB( dstGrid ) = (1.0-OMEGA)*SRC_EB( srcGrid ) + ...
@$\blacktriangleright$@ DST_WT( dstGrid ) = (1.0-OMEGA)*SRC_WT( srcGrid ) + ...
@$\blacktriangleright$@ DST_WB( dstGrid ) = (1.0-OMEGA)*SRC_WB( srcGrid ) + ...
...
SWEEP_END // loop exit
\end{lstlisting}
\end{figure}

\subsection{SPEC CPU2017 538.imagick\_r}
\label{subsec:imagick_r}

538.imagick\_r~\cite{SPEC:CPU2017} is applied to create, edit, compose or convert bitmap images. \loadspy{} reports that spatial redundant loads account for 13\% of the total memory loads, of which 24\% are attributed to the variable \texttt{k} at lines 6-9 shown in Listing \ref{lst:imagick_r}.
\texttt{k} is a pointer to the floating-point array {\texttt{values} at line 2 and decrements by one in each iteration. We find that most elements of this array equal zero, causing \texttt{*k} to equal zero in most of loop iterations. 
	
To remove the identity computation on \texttt{*k}, we introduce a conditional check to filter out all zero values. With this optimization, the memory loads and redundant memory loads are reduced by 19\% and 51\%, and the whole program achieves a 1.25$\times$ speedup.

\begin{figure}[t]
\begin{lstlisting}[firstnumber=1,language=c, caption= {Spatial load redundancy in SPEC CPU2017 538.imagick\_r. Array \texttt{values} is frequently loaded from memory. However, most array elements equal zero.}, label=lst:imagick_r]
register const double *restrict k;
k = &kernel->values[kernel->width*kernel->height-1]
...
for (u=0; u < (ssize_t) kernel->width; u++, k--) {
  if (IsNaN(*k)) continue;
@$\blacktriangleright$@  result.red += (*k)*k_pixels[u].red;
@$\blacktriangleright$@  result.green += (*k)*k_pixels[u].green;
@$\blacktriangleright$@  result.blue += (*k)*k_pixels[u].blue;
@$\blacktriangleright$@  result.opacity += (*k)*k_pixels[u].opacity;
   ... 
}
\end{lstlisting}
\end{figure}

\subsection{Rodinia-3.1 Srad}
\label{subsec:srad}
 
Srad~\cite{rodinia} applies partial differential equations to filter noise in images, which is widely used in ultrasonic and radar imaging applications. We profile the OpenMP version of srad\_v1. \loadspy{} reports a temporal redundancy fraction of 99\%. 8\% of the redundancy is attributed to the array \texttt{image} at lines 12-14 shown in Listing~\ref{lst:srad}. We notice that when 0 $<$ \texttt{i} $<$ \texttt{Nr} - 1,  the value of \texttt{image[iS[i] + Nr*j]} in one iteration equals the value of \texttt{image[k]} in the next iteration and further equals the value of \texttt{image[iN[i] + Nr*j]} in the iteration after next. 

To fix this problem, we adopt scalar replacement to avoid redundant loads across iterations, which eliminates 33\% of the memory loads and yields a 1.11$\times$ speedup for the whole program. 
It is worth noting that the indirect accesses in this inefficient code snippet introduce challenges in compiler's static analysis and optimization. 

Additionally, \loadspy{} also identifies the same inefficiency occurring in srad\_v2. With the same optimization, srad\_v2 achieves a 1.12$\times$ speedup.
 
\begin{figure}[t]
\begin{lstlisting}[firstnumber=1,language=c, caption= Temporal load redundancy in Rodinia-3.1 srad\_v1. Array \texttt{image} is repeatedly loaded from memory while the values remain unchanged., label=lst:srad]
for (i=0; i<Nr; i++) {
  iN[i] = i-1;													
  iS[i] = i+1;													
}
...
iN[0]    = 0;                                                           
iS[Nr-1] = Nr-1;
...
for (j=0; j<Nc; j++) {											
  for (i=0; i<Nr; i++) {											 
    k = i + Nr*j;											
@$\blacktriangleright$@   Jc = image[k];												
@$\blacktriangleright$@   dN[k] = image[iN[i] + Nr*j] - Jc;							
@$\blacktriangleright$@   dS[k] = image[iS[i] + Nr*j] - Jc;							
  }
}
\end{lstlisting}
\end{figure}

\subsection{Rodinia-3.1 LavaMD}
\label{subsec:lavaMD}

LavaMD~\cite{rodinia} calculates particle potential and relocation among particles. We apply \loadspy{} to evaluate its OpenMP version. \loadspy{} reports that 87\% of memory loads are redundant, and the top contributor is the \texttt{glibc} function \texttt{exp} at line 7 shown in Listing~\ref{lst:lavaMD}. 
We notice that the value of \texttt{u2} often remains unchanged across iterations. As a result, a number of redundant loads and computations occur inside \texttt{exp} due to redundant function calls. 
With further analysis, we find that \texttt{a2} is a loop invariant, and \texttt{u2} is derived from \texttt{a2} and \texttt{r2}. Thus, we infer that \texttt{r2} often has the same value across iterations. 

To optimize this inefficiency, we introduce a conditional check on \texttt{r2} such that the program can reuse the return value of function \texttt{exp} from the previous iteration if the value of \texttt{r2} has not changed. This optimization eliminates 76\% of the memory loads and 93\% of the redundant memory loads, yielding a 1.39$\times$ speedup for the entire program.

\begin{figure}[t]
\begin{lstlisting}[firstnumber=1,language=c, caption=Temporal load redundancy in Rodinia-3.1 lavaMD due to redundant function calls., label=lst:lavaMD]
for (k=0; k<(1+box[l].nn); k++) {
  ...
  for (i=0; i<NUMBER_PAR_PER_BOX; i=i+1) {
    for (j=0; j<NUMBER_PAR_PER_BOX; j=j+1) {
      r2 = rA[i].v + rB[j].v - DOT(rA[i],rB[j]);
      u2 = a2*r2;
@$\blacktriangleright$@     vij= exp(-u2);
      fs = 2.*vij;
      ... 
    }
  }
}
\end{lstlisting}
\end{figure}

\section{Threats to validity}
\label{subsec:threat}
The threats mainly exist in applying \loadspy{} for code optimization. 
The same optimization for one application may show different speedups on different computer architectures. 
A given load redundancy fraction may not help estimate the potential speedup.
Some optimizations are input-specific, and a different profile may demand a different optimization. 
Based on the reported inefficiencies, programmers need to devise an optimization that is safe in any execution.

\section{Conclusions}
\label{sec:conclusion}

In this paper, we presented a study of identifying program inefficiencies by focusing on whole-program load redundancy. 
We demonstrate that redundant load operations are often a symptom of various inefficiencies arising from inputs, suboptimal data structure and algorithm choices, and missed compiler optimizations. 
To pinpoint these inefficiencies in complex software code bases, we have developed \loadspy{}, a fine-grained profiler that profiles load redundancy. \loadspy{} toolchain provides valuable guidance to developers for code tuning---calling contexts of the two parties involved in a redundancy, narrowed-down redundancy scopes to focus on optimization, metrics to understand relative significance of redundancy, and a GUI for the source code attribution. 
We evaluate \loadspy{} using several benchmarks and real-world applications. 
Guided by \loadspy{} we are able to optimize prior-known and new inefficiencies in several programs.
Eliminating temporal and spatial load redundancies resulted in nontrivial speedups.

\section*{Acknowledgment}
We thank reviewers for their valuable comments. This work is supported by Google Faculty Research Award and National Natural Science Foundation of China (No. 61502019).

\balance{}
\bibliography{IEEEabrv,references}

\end{document}